\documentclass[journal, 12]{IEEEtran}
\usepackage{amsmath}
\usepackage{amsfonts}
\usepackage{amssymb}
\usepackage{mathrsfs}
\usepackage{amsthm}
\usepackage{cite}
\usepackage{amsbsy}
\usepackage{graphicx}
\usepackage{color}
\usepackage{bm}
\usepackage{algpseudocode}
\usepackage{balance}
\usepackage{mathtools}
\usepackage{xfrac}

\usepackage[nottoc]{tocbibind} 

\usepackage{cleveref}

\newtheorem{thm}{Theorem}
\newtheorem{lem}{Lemma}

\newtheorem{rem}{Remark}

\newcommand\mystack[2]{\genfrac{}{}{0pt}{}{#1}{#2}} 
\begin{document}
{\onecolumn{This paper was submitted for publication in IEEE Transactions on  Communications
on July 11, 2019.

\noindent\textcopyright 2019 IEEE. Personal use of this material is permitted.  Permission from IEEE must be obtained for all other uses, in any current or future media, including reprinting/republishing this material for advertising or promotional purposes, creating new collective works, for resale or redistribution to servers or lists, or reuse of any copyrighted component of this work in other works.}}

\newpage
\twocolumn

\title{Rate Analysis of Cell-Free Massive MIMO-NOMA With Three Linear Precoders}%

\author{Fatemeh~Rezaei,~\IEEEmembership{Student Member,~IEEE,}
       Chintha~Tellambura,~\IEEEmembership{Fellow,~IEEE,}
       Aliakbar~Tadaion,~\IEEEmembership{Senior Member,~IEEE,}
       and~Ali Reza~Heidarpour,~\IEEEmembership{Student Memeber,~IEEE}

\vspace{-3ex}}
\maketitle

\begin{abstract}
Although the  hybrid  of cell-free (CF) massive multiple-input multiple-output (MIMO)  and non-orthogonal multiple access (NOMA) promises  massive spectral efficiency gains,  the  type of  precoders employed at  the access points (APs) impacts  the gains.    In this paper,  we thus  comprehensively  evaluate the system  performance with  maximum ratio transmission (MRT), full-pilot zero-forcing (fpZF) and  modified regularized ZF (mRZF) precoders. We  derive their  closed-form  sum rate expressions by considering  Rayleigh fading channels,  the effects of intra-cluster pilot contamination, inter-cluster interference, and imperfect successive interference cancellation (SIC). Our results reveal that this system  supports   significantly more  users  simultaneously at the same coherence interval compared to  its OMA equivalent. However, intra-cluster pilot contamination and imperfect SIC degrade the system  performance  when the  number of users is low.  Moreover, with  perfect SIC, mRZF and fpZF significantly outperform MRT. Also, we show that this system   with either mRZF or fpZF precoding  outperforms  OMA systems with MRT. The analytical findings are verified by  numerical results. 
\end{abstract}


\IEEEpeerreviewmaketitle
\section{Introduction}
\subsection{Background and Scope}
\IEEEPARstart{I}{nternational} mobile telecommunications (IMT)-$2020$ standard is envisioned to support  peak data rates up to  $10$ Gb/s for low mobility users and  $1$ Gb/s for high mobility users, $1$ ms over-the-air latency, and a network energy
efficiency improvement of 100$\times$ of 4G networks \cite{IMT}. These requirements have necessitated  the  development of  new wireless  technologies  that will offer higher data rates, higher energy efficiency, and ultra low latency.

One such technology is the use of  distributed large antenna arrays, which have the   ability to increase ergodic sum rate \cite{Matthaiou2013}, extend coverage \cite{Wang2013}, and improve energy efficiency \cite{He2016}.  Thus,  distributed  access points (APs) in each cell of a cellular network  have been already studied  \cite{Matthaiou2013, Huang2011}. The use of distributed APs enables the  efficient utilization of spatial resources \cite{Wang2015}. However, inter-cell interference is  inherent   in all  cell-centric systems and becomes a major performance limiting factor  \cite{Interdonato2019}. To overcome this, while preserving the main benefits  of massive MIMO (m-MIMO), cell-free (CF) mMIMO has been proposed \cite{Ngo2017, Nayebi2017}. The basic premise of it is that a large number of spatially-distributed APs  serve many single-antenna users in the same time-frequency resources \cite{Ngo2017}. Thus,  each user  is served by all the APs it can reach, and hence, it does not  experience cell boundaries (cell-free).  A  central processing unit (CPU) coordinates the APs,  which  are connected to it through a fronthaul network. Locality of  the operations of each AP is another key  idea, which  minimizes the fronthaul overhead. Thus, each AP performs precoding based on only the estimates  of the  channels from itself to all the users.  These  downlink channels can be estimated with the help of pilots transmitted by the users in the uplink, thereby exploiting  the   channel reciprocity inherent in time-division duplexing (TDD). This architecture offers  increased macro-diversity and favorable propagation with negligible inter-user interference.   And it   outperforms the  conventional cellular counterparts  by exploiting the best of both collocated mMIMO and network MIMO systems \cite{Ngo2017, Nayebi2017, Nguyen2017, Ngo2018, Attarifar2019}.

Precoder  design, power control, hardware impairments and other factors for CF mMIMO systems have  been investigated.   For instance, \cite{Ngo2017} and \cite{Nayebi2017}  investigate  conjugate beamforming (CB) (which maximizes the signal gain at the intended user), and zero-forcing (ZF) (which nulls  the inter-user interference at the expense of gain losses), and they show that  significant  spectral efficiency  gains are achievable  over conventional   small-cell systems. The problem of joint power control and load balancing using ZF and maximum ratio combining (MRC) processing  is investigated \cite{Nguyen2018}. How the spectral efficiency is impacted by  hardware impairments at both users and APs  has  also been  studied \cite{Zhang2018}; the key insight is that the negative  effects of hardware impairment vanish as  the number of APs  increases.

Non-orthogonal multiple access (NOMA) also  achieves   high  spectral efficiency gains\cite{Islam2017, Ali2017}. It however is a paradigm shift from   conventional orthogonal  multiple access (OMA) techniques such as time division multiple access and frequency division multiple access, where an orthogonal channel must be  created for each different user. In contrast,  NOMA serves   multiple users  simultaneously on  each orthogonal channel. This can be done by exploiting channel gain disparities and then performing individual user signal detection via successive interference cancellation (SIC)  \cite{Islam2017}.

The hybrid of  \textbf{co-located} mMIMO and NOMA offers  a great potential to support low latency massive connectivity requirements of the next-generation wireless networks while further improving the spectral efficiency of NOMA-based systems \cite{Chen2018, Kudathanthirige2019, Senel2019}. Thus, integrating NOMA with  CF mMIMO  may reap further gains and is therefore a  critically  important research topic.
\subsection{Problem Statement and Contributions}Hybrid   CF massive MIMO-NOMA  offers  tremendous potential to  improve the spectral efficiency. To achieve this goal,  two key  factors are essential.  First, precoding of  superposition-coded signals to multiple user-clusters is necessary. Second, the user multiplexing  in each cluster  must be on the  basis of their   channel power gain differentials.   However,  so far  such hybrid  designs have been investigated  by only few  papers \cite{Li2018,Li2019,Bashar2019}.  Thus, the  main objective of this paper is to add to this  literature and further investigate these systems.      

In particular,   for CF massive MIMO-NOMA systems, we investigate   three  linear  precoding schemes, with the same front-hauling overhead. These are  maximum ratio transmission (MRT), full-pilot zero-forcing (fpZF)  and modified regularized ZF (mRZF). Of these three, while its performance in mMIMO systems has been investigated  \cite{Wagner2012, Huh2012,  Rezaei2019}, mRZF has not been studied for CF  massive MIMO before.   The  advantages of mRZF  include  its ability to balance the interference suppression and desired signal power and also having additional parameters to be optimized.  

Normal  ZF and RZF precoders require exchanging  instantaneous channel state information (CSI) among the APs -- the main benefit of  mRZF and fpZF is the elimination of  this exchange  and that each AP computes  its precoder with only its local CSI. This is a big advantage as mRZF and fpZF have the same front-hauling overhead as  MRT \cite{Chien2016,Interdonato2018}. Unlike MRT which maximizes the signal gain at the intended cluster and ignores the inter-cluster interference, fpZF sacrifices some of the array gain to cancel the inter-cluster interference  \cite{Chien2016,Interdonato2018}. On the other hand, mRZF balances the inter-cluster interference  mitigation and intra-cluster power enhancement.  Note that the   optimal precoding is  nonlinear, so these linear precoders are  sub-optimal. But they offer  high performance with affordable computational complexity, making them ideal for  practical large  MIMO systems \cite{Wagner2012}.

The main contributions of this paper on CF massive MIMO-NOMA systems  can be summarized as follows:
\begin{enumerate}
    \item We derive the  closed-form  downlink  sum rate   when the APs  employ  MRT or  fpZF precoders, considering the effects of intra-cluster pilot contamination, and  inter-cluster interference. The users rely on the statistics for their effective channels for decoding, and the SIC process is imperfect. 
    \item We also analyze the performance when the  APs employ  mRZF precoding. Since the closed-form analysis of it  with finite system parameters is  difficult (if not impossible), we analyze the achievable rate when  the number of clusters ($N$) and the number of antennas at each AP ($L$) grow infinitely large while  $ \frac L N $  is a  finite ratio.
  
  \item  We  show that NOMA allows  a significant number of users to be  supported simultaneously at the same coherence interval compared to its counterpart OMA. For instance, with  $K$ users in each cluster,  NOMA based CF massive MIMO can support $K$ times the users that of OMA. Moreover, for a large number of users, NOMA outperforms OMA; however, for  low number of users, the sum rate of  NOMA is lower than that of OMA due to the effects of intra-cluster pilot contamination and imperfect SIC. It is further shown that given perfect SIC, mRZF and fpZF significantly outperform MRT. We also show that  either mRZF or fpZF outperforms OMA systems with MRT. 
    \item  Numerical results are also presented to support our findings.
\end{enumerate}
\subsection{Previous Contributions on  CF massive MIMO-NOMA }
 As mentioned before, the investigation of CF massive MIMO-NOMA systems has been sparse, except for  \cite{Li2018,Li2019,Bashar2019}.  In fact \cite{Li2018} is the first paper to study the design of such systems.  It  considers  single-antenna APs that use conjugate beamforming precoders. It thus analyzes degradations due to  estimated/imperfect CSI at APs, SIC, and statistical CSI at users.  On the other hand,  \cite{Li2019} generalizes to   multiple-antenna APs and   derives minimum mean square error (MMSE) uplink  and  downlink channel estimates, and for both these cases, the achievable rate is derived. Reference \cite{Li2019}    also considers  spatial correlation among the multiple antennas  at each AP and intra-cluster pilot contamination and  quantifies  the adverse impact of  intra-cluster pilot contamination and error propagation due to imperfect SIC at the user nodes.  The  only other work  \cite{Bashar2019} deals with  the  max-min fairness based bandwidth efficiency problem. It develops an optimal  algorithm and mode switching  between NOMA and OMA for maximum bandwidth efficiency.  

The above paragraph clarifies  the differences between this paper and  \cite{Li2018,Li2019,Bashar2019}.   The focus of this paper is  to comparatively evaluate the  three types of  practical precoders. We believe that this paper is the first to do that  in the context of CF massive MIMO-NOMA. 

\subsection{General works on MIMO NOMA} Since the NOMA literature is vast, we mention only few works here. NOMA outperforms  OMA from user and system throughput perspectives  \cite{Benjebbour2014}. In \cite{Benjebbour2013},  multi-user power allocation, user scheduling, and error propagation in SIC for NOMA are  investigated. MIMO-NOMA outperforms conventional MIMO-OMA systems \cite{Zeng2017,Ding2016}. MIMO-NOMA outperforms  MIMO-OMA  in terms of both sum channel capacity and ergodic sum capacity \cite{Zeng2017}. Furthermore, the user outage probability in a MIMO-NOMA cluster is studied in \cite{Ding2016}. Beamforming, user clustering, and power allocation for MIMO-NOMA are further investigated in \cite{Huang2018}.

Hybrid massive MIMO-NOMA,  which outperforms  standalone mMIMO and NOMA schemes, has been studied in  \cite{Senel2019, Vaezi2019}. The outage probability with perfect user ordering and limited feedback is investigated in \cite{PDing2016}.  In \cite{Chen2018}, a user pairing and pair scheduling algorithm is proposed to enhance the spectral efficiency of massive MIMO-NOMA systems. Moreover, in \cite{Kudathanthirige2019}, user clustering and pilot assignment
schemes for multi-cell massive MIMO-NOMA are investigated by employing the characteristics of correlated fading channels.
\subsection{Structure and Notations}
This  paper is organized as follows. The system model is introduced in Section \ref{system_modelA}. In Section \ref{acheivedRate}, the achievable sum rate with linear precoding techniques is derived. In Section \ref{sim}, the analytical
results are confirmed through simulation examples. Section \ref{conclusion} concludes the paper.

\textit{Notation}: Lower-case bold  and upper-case bold denote vectors and matrices, respectively. $\mathbf{I}_n$ represents the $n \times n$ identity matrix. $\mathbf{A}^\mathrm{T}$, $\mathbf{A}^\mathrm{H}$, $\text{tr}[\mathbf{A}]$, and $[\mathbf{A}]_{(m,n)}$ denote transpose,  Hermitian transpose, trace, and the $(m,n)$th element of  matrix $\mathbf{A}$, respectively.  $\mathbb{E} \{ \cdot\}$ denotes the statistical expectation. Finally, $\mathcal{CN}(\bm{\mu},\mathbf R ) $ is a complex Gaussian  vector with  mean   $\bm \mu$ and co-variance matrix  $\mathbf R$.
\section{System Model and Preliminaries}\label{system_modelA}
Here,  we describe the system model, channel model, and transmission model in detail.
\begin{figure}
\centering
\includegraphics[width=9cm, height=7.5cm]{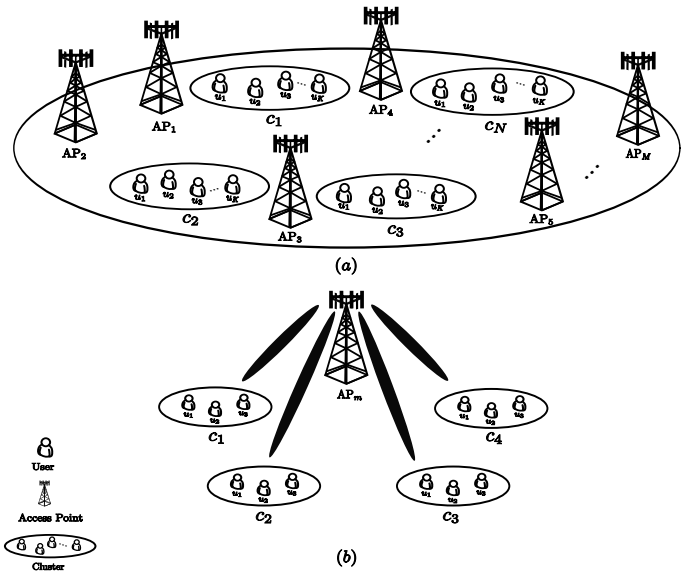}
\caption{ (a) System model of cell-free massive MIMO-NOMA  (b) An AP transmits  to four  clusters  $\{ c_1, \ldots, c_4\}$, each with  three  users.}\label{system_model}
\vspace{-2ex}
\end{figure}
\subsection{System and Channel Models}
Consider the downlink transmission of  CF massive MIMO-NOMA. This system has  $M$  APs   and $KN$ single-antenna users, which are grouped into $N$ clusters with $K\; (K \ge 2)$ users per cluster and NOMA is applied among the users in the same cluster. Each AP is equipped with $L$ antennas (Fig.~\ref{system_model}.a). The total  $M$ APs  serve $KN$ users in the same time-frequency resource block, where $KN\ll ML$. All  the  APs  are connected to a  CPU  via an error-free  fronthaul network to achieve coherent processing  \cite{Ngo2017}. This network  carries only  the payload data,  large-scale parameters and power control coefficients that change slowly. Each  AP   computes its  precoder  based on the estimates of the channel states between itself and the users. Importantly,  it does not need the knowledge of the channel states between the users and other APs.   

The downlink channel between  the $m$th AP ($m=1, \ldots, M$)  and the $k$th user ($k=1,\ldots, K$)  in the $n$th cluster ($n=1,\ldots, N)$ is the  complex Gaussian random vector 
\begin{equation}\label{channel}
\mathbf{h}_{mnk} \sim  \mathcal{CN}\left(\mathbf{0}, \beta_{mnk}\mathbf{I}_L\right),
\end{equation}
where $\{\beta_{mnk}\}$ are the set of  the large-scale fading coefficients.  Each  AP is assumed  to   know its own set of  large-scale coefficients \cite{Ngo2017}. This assumption is justifiable because  they    are quasi-static  and hence
only need to be estimated once about every 40 coherence
time intervals \cite{tse2004}.  This model \eqref{channel} also  presumes that small-scale fading is    Rayleigh distributed. 
\subsection{Uplink Pilot Transmission}
The users transmit   pilot sequences in the uplink,  and  the APs use them to estimate the downlink  channels.  In order to minimize the channel estimation overhead in NOMA,  the  same  pilot  sequence of length $\tau$ samples, is  shared  among  the  users  within  each  cluster,  but  different clusters are assigned mutually orthogonal pilots\footnote{It is different from OMA where all the users in all the clusters are assigned with different mutually orthogonal pilots.}\cite{Li2018}. Then, the pilot sequence for the $k$th user in the $n$th cluster is $\sqrt{\tau}\bm{\phi}_{n} \in \mathcal{C}^{\tau \times 1}$ satisfying $\|\bm{\phi}_{n}\|^2 = 1$ where $N\leq \tau$. Accordingly, the pilot sequences allocated to the $N$ clusters are mutually orthogonal implying that $\bm{\phi}_{n}^{\mathrm{H}} \bm{\phi}_{j} = 0$ $(n \neq j)$. The received pilot signal at the $m$th AP ($\mathbf{Y}_{m}^p \in \mathcal{C}^{L \times \tau}$) can then be expressed as
\begin{equation}\label{pilot_signal}
\mathbf{Y}_{m}^p = \sqrt{\tau p_p} \sum\limits_{n=1}^N{ \sum\limits_{k=1}^K{\mathbf{h}_{mnk}\bm{\phi}_{n}^\mathrm{H}}+ \mathbf{N}_m}, \quad \forall m
\end{equation} 
where $p_p$ is the pilot transmit power and $\mathbf{N}_m \in \mathcal{C}^{L \times \tau}$ is a Gaussian noise matrix with i.i.d $\mathcal{CN}(0, 1)$ elements. Then, the $m$th AP estimates $\mathbf{h}_{mnk}$ using minimum mean square error (MMSE) estimation. To do so, the received pilot signal at the $m$th AP \eqref{pilot_signal} is projected onto ${\bm{\phi}}_{n}$ which yields
\begin{equation}\label{projection}
\mathbf{\tilde{y}}_{mn}^p = \sqrt{\tau p_p} \sum\limits_{k=1}^K{\mathbf{h}_{mnk}}+ \tilde{\mathbf{n}}_{mn},
\end{equation}
where $\tilde{\mathbf{n}}_{mn} = \mathbf{N}_m {\bm{\phi}}_{n} \sim \mathcal{CN}(\mathbf{0},\mathbf{I}_L)$.

By using \eqref{projection}, the MMSE estimate  of $\mathbf{h}_{mnk}$ can be expressed as $\hat{\mathbf{h}}_{mnk} =  c_{mnk}\mathbf{\tilde{y}}_{mn}^p$,  where $c_{mnk}$ is given by
\begin{equation}
c_{mnk} = \frac{\sqrt{\tau p_p}\beta_{mnk}}{1+ \tau p_p \sum \limits_{i=1}^K{\beta_{mni}}}.
\end{equation}
Since  $\mathbf{\tilde{y}}_{mn}^p$ is Gaussian distributed, $\hat{\mathbf{h}}_{mnk}$ can be written as,
\begin{equation}\label{estimate}
\hat{\mathbf{h}}_{mnk} = \sqrt{\theta_{mnk}} \mathbf{\nu}_{mn},
\end{equation} 
where $\mathbf{\nu}_{mn} \sim \mathcal{CN}(\mathbf{0}, \mathbf{I}_L)$ and $\theta_{mnk}$ is equal to
\begin{equation}\label{theta}
\theta_{mnk}= \frac{1}{L}\mathbb{E} \left\{\left\| \hat{\mathbf{h}}_{mnk}\right\|^2 \right\} = \frac{\tau p_p\beta_{mnk}^2}{1+ \tau p_p \sum\limits_{i=1}^K{\beta_{mnk}}}.
\end{equation}
The channel estimation error can then be defined as $\bm{\epsilon} _{mnk} = \mathbf{h}_{mnk}- \hat{\mathbf{h}}_{mnk}$ where $\bm{\epsilon}_{mnk} \sim \mathcal{CN}(\mathbf{0}, (\beta_{mnk}- \theta_{mnk})\mathbf{I}_L)$. Since the users in the same cluster use the same pilot sequence, their channels are parallel as shown in \eqref{estimate}. Mathematically, this can be written as
\begin{equation}\label{parallel}
 \hat{\mathbf{h}}_{mnk} = \frac{\beta_{mnk}}{\beta_{mni}}\hat{\mathbf{h}}_{mni},
\end{equation}
where $\beta_{mnk}$ is already given by \eqref{channel}.

\subsection{Downlink Data Transmission Model}
The superposition coded data signal for the $K$ users in the $n$th cluster is expressed as
\begin{equation}\label{downlink_data}
s_n  = \sum\limits_{k = 1}^K {\sqrt {p_{nk}} s_{nk} }, \quad \forall n . 
\end{equation}
In \eqref{downlink_data}, $s_{nk}$ and $p_{nk}$ denote the data signal and the transmitted power allocated to the $k$th user in the $n$th cluster. Also, $p_{nk} = p_{t} \lambda_{nk}$ where $p_{t}$ is the total transmit power of each AP and $\{\lambda_{nk}\}$ are the set of  power coefficients that satisfies $\sum_{n=1}^N{\sum_{k = 1}^K{\lambda_{nk}}} = 1$. Furthermore,  the different  data signals are mutually uncorrelated: 
\begin{equation}\label{Eqn:2}
\mathbb{E}\{s_{nk} s_{mi}^* \}  = \left\{ \begin{array}{l}
 1, \quad m = n \quad \& \quad k = i \\ 
 0,\quad \rm{else} \\ 
 \end{array} \right., 
\end{equation}
where $m,n \in \{ 1,2, \ldots, N\}$ and $k,i \in \{ 1,2, \ldots, K\}$. Therefore,
 $\mathbb{E}\{ |s_n|^2\} =  \sum_{k = 1}^K {p_{nk}} = p_{t} \lambda_n$, where $\lambda_n =  \sum_{k=1}^K{\lambda_{nk}}$ accounts for the power allocation coefficient for the $n$th cluster.

The $m$th AP ($ m=1,\ldots, M$) transmits the signal 
\begin{align}\label{Eqn:transmit}
\mathbf{x}_m = \sum\limits_{n=1}^N{\mathbf{w}_{mn} s_n}, 
\end{align}
where $\mathbf{w}_{mn} \in \mathcal{C}^{L}$ represents the spatial directivity of the signals sent to the users in the $n$th cluster.  Note that each AP  precodes the transmitted signals for all the users in the same cluster with the same beamforming vector $\mathbf{w}_{mn}$, i.e., each AP has $N$  precoding vectors (Fig.~\ref{system_model}.b).

Since the $KN$ users are served simultaneously by $M$ APs, the received signal at the $k$th user in the $n$th cluster can be expressed  as
\begin{align}\label{Eqn:precode2}
\nonumber {y}_{nk} &= \underbrace{\sum\limits_{m = 1}^M {\sqrt{p_{nk}}\mathbf{h}_{mnk}^\mathrm{H}\mathbf{w}_{mn} s_{nk}}}_\text{desired signal} + \underbrace{\sum\limits_{m = 1}^M {\mathbf{h}_{mnk}^\mathrm{H}\mathbf{w}_{mn} \sum\limits_{\mystack{i=1}{ i \ne k} }^K {\sqrt{p_{ni}} s_{ni}}}}_\text{intra-cluster interference before SIC} \\
  &+ \underbrace{\sum\limits_{m = 1}^M {\mathbf{h}_{mnk}^\mathrm{H} \sum\limits_{\mystack{n'=1}{n' \ne n}}^N{\mathbf{w}_{mn'} s_{n'}}}}_\text{inter-cluster interference} + {n}_{nk}.
\end{align}
In \eqref{Eqn:precode2}, ${n}_{nk} \sim \mathcal{CN}({0},1)$ while the first, second, and the third terms are desired signal, intra-cluster interference before SIC and inter-cluster interference, respectively.

In TDD CF massive MIMO systems, with sufficiently large number antennas at each AP (e.g., $5-10$ antennas per AP \cite{zChen2018}), the instantaneous channel coefficients can be well approximated by their corresponding expected values. This phenomenon is referred to as ``\textit{channel hardening}" which substantially reduces the effective channel gain fluctuations; thereby rendering the use of downlink pilot training unnecessary \cite{Ngo2017, zChen2018}.  To apply power domain NOMA, we thus assume that the users in the $n$th cluster are ordered based on the mean of the effective channel gains as follows \cite{Li2018,Li2019}:

\begin{equation}\label{Eqn:NOMA}
\Gamma_{n1} \geq \Gamma_{n2}\geq  \cdots \geq \Gamma_{nK},
\end{equation}
where  $\Gamma_{nk} = \mathbb{E}\left\{\left|\sum_{m = 1}^M {\hat{\mathbf{h}}_{mnk}^\mathrm{H}\mathbf{w}_{mn}}\right|^2 \right \}$.
 
The user ordering can be done by a central entity (e.g., the CPU) that collects all information of the mean of effective channel strengths and then feeds back the result of user ordering  and power allocation to the APs; higher powers are allocated to the users with lower channel strength; i.e., $p_{n1} \leq p_{n2} \leq \cdots \leq p_{nK}$. Hence, the $k$th user applies SIC to decode its own signal. More precisely, it decodes the signals of the users with higher powers and treats the others as interference. In particular, the $k$th user decodes the signal intended for the $i$th user ($\forall i\geq k$) and treats the signals of the other users ($\forall i<k$) as interference.

Using the statistical CSI knowledge of the effective channels at users, i.e., $\mathbb{E} \left\{ \mathbf{h}_{mnk}^\mathrm{H} \mathbf{w}_{mn}\right\} \; (\forall m,n,k$),  \eqref{Eqn:precode2} can be rewritten as follows:
\begin{align}\label{Eqn:sic}
\nonumber {y}_{nk} &= \underbrace{\sum\limits_{m = 1}^M {\sqrt{p_{nk}} \mathbb{E} \left\{\mathbf{h}_{mnk}^\mathrm{H}\mathbf{w}_{mn} \right\} s_{nk}}}_{T_0: \text{desired signal}} \\ 
\nonumber&+ \underbrace{\sum\limits_{m = 1}^M {\sqrt{p_{nk}} \left(\mathbf{h}_{mnk}^\mathrm{H}\mathbf{w}_{mn} - \mathbb{E} \left\{\mathbf{h}_{mnk}^\mathrm{H}\mathbf{w}_{mn}\right\}\right) s_{nk}}}_{T_1: \text{beamforming gain uncertainty}}  \\
\nonumber &+\underbrace{\sum\limits_{m = 1}^M {\sum\limits_{i = 1}^{k-1}{ \sqrt{p_{ni}} \mathbf{h}_{mnk}^\mathrm{H}\mathbf{w}_{mn}s_{ni} }}}_{T_2: \text{intra-cluster interference after SIC}}  \\ 
\nonumber &+ \underbrace{\sum\limits_{m = 1}^M {\sum\limits_{i = k+1}^K{ \sqrt{p_{ni}} \left(\mathbf{h}_{mnk}^\mathrm{H}\mathbf{w}_{mn} s_{ni} - \mathbb{E} \left\{\mathbf{h}_{mnk}^\mathrm{H}\mathbf{w}_{mn}\right\} \hat{s}_{ni}\right)}}}_{T_3: \text{residual interference due to imperfect SIC}}\\
 &+ \underbrace{\sum\limits_{m = 1}^M {\mathbf{h}_{mnk}^\mathrm{H} \sum\limits_{\mystack{n' = 1}{n' \ne n}}^N{\mathbf{w}_{mn'} s_{n'}}}}_{T_4: \text{Inter-cluster interference}} + {n}_{nk},
\end{align}
where the first term $T_0$ is the desired signal and the second term $T_1$ is the beamforming gain uncertainty. The  third term $T_2$ is intra-cluster interference caused by the signals of the users which are considered as the interference at the $k$th user. Note that because the users know  statistical CSI only, intra-cluster pilot contamination and channel estimation error, SIC may not be perfect. The forth term $T_3$ thus represents the error propagation due to the imperfect SIC where $\hat{s}_{ni}$ is the estimate of $s_{ni}$. Here, we adopt the linear MMSE estimation \cite{Kay1993} where $s_{ni}$  and its estimate $\hat{s}_{ni}$ can be assumed as jointly Gaussian distributed with a normalized correlation coefficient. The relationship between $s_{ni}$ and $\hat{s}_{ni}$ can then be written as

\begin{equation}\label{Eqn:decod}
s_{ni} = \rho_{ni} \hat{s}_{ni} + e_{ni},
\end{equation}
where $\hat{s}_{ni}\sim \mathcal{CN}({0},1) $, $e_{ni} \sim \mathcal{CN}({0},\sigma^2_{e_{ni}}/[1+\sigma^2_{e_{ni}}])$ is the estimation error, statistically independent of $\hat{s}_{ni}$, and  $ \rho_{ni} =1 / \sqrt{1+\sigma^2_{e_{ni}}}$ . The correlation coefficient $0 \leq \rho \leq 1$  reflects the quality of the estimation and quantify the severity of the SIC imperfection. The value of $\rho$ is determined by channel related  issues  (fading  and  shadowing) and other factors \cite{Li2011}; the greater its value, the greater the correlation between $\hat{s}_{ni}$ and $s_{ni}$ and  better the SIC performance \cite{Li2018}.

\section{Downlink Achievable Rate}\label{acheivedRate}
Here, we analyze  the downlink sum rates achievable with  MRT, fpZF, and mRZF. To this end, we first derive  the effective signal-to-interference-plus-noise ratio (SINR) at each user. It then leads to the total downlink sum rate.  

According to \eqref{Eqn:sic}, since data signals intended to different users are mutually uncorrelated and the white additive noise is independent from the data symbols and the channel coefficients, it is easy to check that $T_i, \forall i$ and $n_{nk}$ are mutually uncorrelated. Therefore, by considering the first term in \eqref{Eqn:sic} as the desired signal and the remaining terms as an effective noise, invoking the argument in \cite{Marzetta2016}, the SINR at the $k$th user in the $n$th cluster $\gamma_{nk}$ can be derived as \eqref{Eqn:INTERFERENCE}, shown at the top of the next page, where $\eta_{nk}$ is defined as $\eta_{nk} \triangleq \sum_{m = 1}^M {\mathbf{h}_{mnk}^\mathrm{H}\mathbf{w}_{mn}}$. 
\begin{figure*}
\begin{equation}\label{Eqn:INTERFERENCE}
\gamma_{nk} = \frac{p_{nk} \left|\mathbb{E} \left\{\eta_{nk} \right\}\right|^2 }
{p_{nk} \mathbb{E} \left\{\left|\left(\eta_{nk} - \mathbb{E} \left\{\eta_{nk}\right\}\right)\right|^2\right\}+\sum\limits_{i = 1}^{k-1}{p_{ni}\mathbb{E}\left\{|\eta_{nk}|^2\right\}}+\sum\limits_{i = k+1}^K { p_{ni}\mathbb{E} \left\{\left| \eta_{nk} s_{ni} - \mathbb{E} \left\{\eta_{nk}\right\} \hat{s}_{ni}\right|^2 \right\}}+\sum\limits_{\mystack{n'=1}{n' \ne n}}^N{p_{n'} \mathbb{E} \left\{ \left|\eta_{n'k}\right|^2 \right\}}+1}.
\end{equation}
\rule{\textwidth}{0.25mm}
\end{figure*}
The achievable rate for the $k$th user in the $n$th cluster can then be computed as
\begin{equation}\label{Eqn:rate}
\mathcal{R}_{nk} = \zeta {\rm{log}}_2(1+\gamma_{nk}),
\end{equation}
where $\zeta = {(\tau_c -\tau)}/{\tau_c}$ is the pre-log factor and $\tau_c$ is the coherence interval.

The achievable rate $\mathcal{R}_{nk}$ in \eqref{Eqn:rate} depends on the precoding process at the APs. In order to design the precoding matrices at the APs, we first define  $\bar{\mathbf{H}}_m  \triangleq \mathbf{Y}_{m}^p \boldmath{\Phi} $ where $\bm{\Phi} =[\bm{\phi}_1,\bm{\phi}_2, \ldots,\bm{\phi}_N] \in \mathcal{C}^{\tau \times N}$ and $\mathbf{Y}_{m}^p \in \mathcal{C}^{L \times \tau}$ is given by \eqref{pilot_signal}. $\bar{\mathbf{H}}_m$ can be then written as

\begin{equation}\label{Eqn:h_bar}
\bar{\mathbf{H}}_m = \left[\sqrt{\tau p_p} \sum\limits_{k=1}^K{\mathbf{h}_{m1k}}+ \tilde{\mathbf{n}}_m,  \ldots, \sqrt{\tau p_p} \sum\limits_{k=1}^K{\mathbf{h}_{mNk}}+ \tilde{\mathbf{n}}_m\right].
\end{equation}
Therefore,  $\bar{\mathbf{H}}_m$ has independent columns and $\bar{\mathbf{h}}_{mn} \sim \mathcal{CN}(\mathbf{0}, (1 +\tau p_p \sum\limits_{k=1}^k{\beta_{mnk}})\mathbf{I}_L)$ is the $n$th column of $\bar{\mathbf{H}}_m$ i.e., $\bar{\mathbf{h}}_{mn} = \bar{\mathbf{H}}_m \mathbf{e}_n$ where  $\mathbf{e}_{n}$ denotes the $n$th column of $\mathbf{I}_N$. Then, the channel estimate can be written as
\begin{equation}\label{Eqn:estimate}
\hat{\mathbf{h}}_{mnk} = c_{mnk} \bar{\mathbf{h}}_{mn}. 
\end{equation}

We next  assume that the APs use either MRT, fpZF or mRZF to precode data signals. For each precoder,  we then derive the closed-form achievable sum rate.
\subsection{MRT beamforming}
With  MRT, the $m$th AP computes the following  precoding vector   for the $k$th user in the $n$th cluster: 
\begin{equation}\label{Eqn:mrt}
\mathbf{w}_{mn} = \frac{\bar{\mathbf{h}}_{mn}}{\sqrt{\mathbb{E} \left\{ \parallel \bar{\mathbf{h}}_{mn} \parallel^2\right\}}}.
\end{equation}

This precoder  \eqref{Eqn:mrt} is used for  all  the users in the $n$th cluster. 
\begin{thm}
In the CF massive MIMO-NOMA system with MRT precoding, $\gamma_{nk}$ in \eqref{Eqn:INTERFERENCE}, for finite values of $M, L ,N$ and $K$, is given by \eqref{Eqn:mrt_rate}, as shown at the top of the next page.
\end{thm}
\begin{figure*}
\begin{equation}\label{Eqn:mrt_rate}
 \gamma_{nk}^{\rm{MRT}} = \frac{L p_{nk} \left(\sum \limits_{m=1}^M{\sqrt{\theta_{mnk}}}\right)^2}{ L \left(\sum \limits_{m=1}^M{\sqrt{\theta_{mnk}}}\right)^2 \left(\sum\limits_{i=1}^{k-1}{p_{ni}} + \sum\limits_{i=k+1}^K{p_{ni} (2-2 \rho_{ni})}\right)+\left(\sum\limits_{n'=1}^N{p_{n'}}\right) \sum\limits_{m=1}^M{\beta_{mnk}}+1}.
\end{equation}
\rule{\textwidth}{0.25mm}
\end{figure*}
\begin{proof}
See Appendix \ref{mrt_app} for derivations.
\end{proof}
Since   \eqref{Eqn:mrt_rate} connects several factors together, it can be used to   provide several  remarks and guidelines  relevant for practical  NOMA-based  CF mMIMO.  We  briefly describe them  next. 

\begin{rem}\label{REM1_MTR}
With  MRT precoding, the signal power increases as the number of antennas at each AP,  $L,$ increases, thanks to the array gain. On the other hand, the interference due to the pilot contamination and imperfect SIC (the first term in the denominator) proportionally increases as $L$ increases.
\end{rem}
\begin{rem}\label{REM2_MTR}
Beamforming gain uncertainty and inter-cluster interference (the second term in the denominator), are  not affected by $L$.
\end{rem}
\begin{rem}\label{REM3_MTR}
As opposed to the number of clusters $N$, which is restricted by the maximum orthogonal pilot sequence length ($ N\leq \tau$), the number of users $K$ within each cluster can be increased greatly.   However, the downside is the  increase of pilot contamination and imperfect SIC.
\end{rem}
  
\subsection{Full-pilot ZF beamforming}
In this case, the precoding vector at the $m$th AP for the users in the $n$th cluster can be formulated as 
\begin{equation}\label{Eqn:precode}
\mathbf{w}_{mn} = \frac{\bar{\mathbf{H}}_m \left(\bar{\mathbf{H}}_m^\mathrm{H} \bar{\mathbf{H}}_m\right)^{-1}\mathbf{e}_{n}}{\sqrt{\mathbb{E} \left\{ \left\lVert  \bar{\mathbf{H}}_m (\bar{\mathbf{H}}_m^\mathrm{H} \bar{\mathbf{H}}_m)^{-1}\mathbf{e}_{n} \right \rVert^2 \right\}}},  
\end{equation}
where each AP has $N$ precoding vectors, one per cluster (pilot).

In order to find  the effective SINR $\gamma_{nk}$ given in \eqref{Eqn:INTERFERENCE}, we first need to derive the value of $\eta_{nk}, \forall n$.  To obtain $\eta_{nk}, \forall n$, we need to compute $\hat{\mathbf{h}}_{mnk}^\mathrm{H}\mathbf{w}_{mn'}$. To this end, we first obtain the normalization term in \eqref{Eqn:precode}. By employing Lemma $2.10$ of \cite{tulino2004} and using the properties of $N \times N$ central complex Wishart matrix with $L$ ($L\geq N+1$) degrees of freedom, we obtain
\begin{align}\label{Eqn:INTERFERENCE2}
\nonumber &\mathbb{E} \left\{ \parallel \bar{\mathbf{H}}_m (\bar{\mathbf{H}}_m^\mathrm{H} \bar{\mathbf{H}}_m)^{-1}\mathbf{e}_{n}\parallel^2\right\} = \mathbb{E}\left\{(\bar{\mathbf{H}}_m^\mathrm{H} \bar{\mathbf{H}}_m)^{-1}_{n,n} \right\}\\
&= \frac{1}{(L-N)(1+ \tau p_p \sum_{k=1}^K{\beta_{mnk}})}.
\end{align}

Finally, by employing \eqref{Eqn:estimate}, \eqref{Eqn:precode} and \eqref{Eqn:INTERFERENCE2}, we have
\begin{align}\label{Eqn:INTER}
\nonumber \hat{\mathbf{h}}_{mnk}^\mathrm{H}\mathbf{w}_{mn'} &= c_{mnk} \mathbf{e}_{n}^\mathrm{H} \mathbf{e}_{n'} \sqrt{(L-N)\left(1+ \tau p_p \sum_{k=1}^K{\beta_{mnk}}\right)} \\
&= \left\{ \begin{array}{l}
\sqrt{\theta_{mnk}(L-N)},\IEEEeqnarraynumspace n = n'\\ 
0,  \!\!\!\!\IEEEeqnarraynumspace \IEEEeqnarraynumspace \IEEEeqnarraynumspace \IEEEeqnarraynumspace \IEEEeqnarraynumspace n \neq n'  \\ 
 \end{array} \right.
\end{align}

From \eqref{Eqn:INTER}, we find that fpZF precoder suppresses the inter-cluster interference by only utilizing local CSI rather than CSI  shared among the APs, a  big advantage. Moreover,   fpZF has the same front-hauling overhead as  MRT.
\begin{thm}
In the CF massive MIMO-NOMA system with fpZF precoder, $\gamma_{nk}$ in \eqref{Eqn:INTERFERENCE} , for any finite $M, L ,N$ and $K$, is given by \eqref{Eqn:rat2} as shown at the top of the next page.
\end{thm}
\begin{figure*}
\begin{equation}\label{Eqn:rat2}
\gamma_{nk}^{{\rm{fpZF}}} = \frac{\left(L-N\right)p_{nk} \left(\sum \limits_{m=1}^M{\sqrt{\theta_{mnk}}}\right)^2 }{(L-N)\left(\sum \limits_{m=1}^M{\sqrt{\theta_{mnk}}}\right)^2 \left(\sum\limits_{i=1}^{k-1}{p_{ni}} + \sum\limits_{i=k+1}^K{p_{ni}(2-2 \rho_{ni})}\right)+\left(\sum\limits_{n'=1}^N{p_{n'}}\right) \sum\limits_{m=1}^M{\left(\beta_{mnk}-\theta_{mnk}\right)}+1}.
\end{equation}
\rule{\textwidth}{0.25mm}
\end{figure*}
\begin{proof}
See Appendix \ref{zf_app} for the derivations.
\end{proof}
\begin{rem}
Similar to MRT, by employing fpZF precoding, the signal power increases as the number of antennas at each AP $L$ increases, thanks to the array gain (which is $L-N$). On the other hand, the interference due to the pilot contamination and imperfect SIC (the first term in the denominator) proportionally increases as $L$ increases.
\end{rem}
\begin{rem}
Unlike MRT which only aims to maximize the signal-to-noise ratio (SNR) and ignores the inter-cluster interference, fpZF aims to suppress the inter-cluster interference by sacrificing  array gain.
\end{rem}

\begin{rem}\label{remark}
For a fixed number of clusters $N$, if the number of antennas at each AP tends to infinity $(L \to \infty )$, the SINR for both the MRT and fpZF precoding \eqref{Eqn:mrt_rate} and \eqref{Eqn:rat2} can be approximated as
\begin{equation}\label{converge}
\small
\mathop{\lim }\limits_{L \to \infty } \gamma_{nk}=\frac{p_{nk} \left(\sum \limits_{m=1}^M{\sqrt{\theta_{mnk}}}\right)^2 }{\left(\sum \limits_{m=1}^M{\sqrt{\theta_{mnk}}}\right)^2 \left(\sum\limits_{i=1}^{k-1}{p_{ni}} + \sum\limits_{i=k+1}^K{p_{ni}(2-2 \rho_{ni})}\right)}. 
\end{equation}
which shows that, the gain of adding more antennas at each AP disappears. Besides, for fixed  values of $L$ and $N$, by increasing the number of users at each cluster, pilot contamination and SIC become the dominant interferences. 
Therefore, the performance of the CF massive MIMO-NOMA system is limited by pilot contamination and imperfect SIC.

In contrast, in  OMA, where each user is assigned with an orthogonal pilot,  SINR  $\gamma_{nk}$ for  MRT and fpZF is given as
\begin{equation}\label{Eqn:mrt_rate_oma}
 \gamma_{nk,\rm{OMA}}^{\rm{MRT}} = \frac{L p_{nk} \left(\sum \limits_{m=1}^M{\sqrt{\theta_{mnk}}}\right)^2}{\left(\sum\limits_{n'=1}^N{\sum\limits_{k'=1}^K{p_{n'k'}}}\right)\sum\limits_{m=1}^M{\beta_{mnk}}+1} 
\end{equation}
and 
\begin{equation}\label{Eqn:ZF_rate_oma}
 \gamma_{nk,\rm{OMA}}^{\rm{fpZF}} = \frac{\left(L-KN\right)p_{nk} \left(\sum \limits_{m=1}^M{\sqrt{\theta_{mnk}}}\right)^2 }{\left(\sum\limits_{n'=1}^N{\sum\limits_{k'=1}^K{p_{n'k'}}}\right)\sum\limits_{m=1}^M{\left(\beta_{mnk}-\theta_{mnk}\right)}+1}.
\end{equation}
Therefore, deploying more antennas at each AP is always beneficial for OMA.
\end{rem}

\subsection{Modified  RZF beamforming}
The mRZF precoder  tries to balance inter-cluster interference mitigation and intra-cluster power enhancement. Since closed-form analysis of it is  difficult (if not impossible),  we analyze its  achievable rate in the asymptotic regime, where the number of clusters $N$ and the number of antennas at each AP $L$ grows infinitely large while keeping a finite ratio; i.e.,
$ 1\le \text{lim}_{L,N\to \infty} \frac L N \le \infty$\footnote{This assumption implies that the coherence time of the channel $\tau_c$ scales linearly with $N$. However, as it will be shown in Section IV, our analysis provides a tight approximation for the achievable rate even for small values of $N$ or equivalently $\tau_c$.}. This asymptotic expression can nevertheless  be used with  finite values of $L$ and $N$ \cite{Wagner2012, Huh2012}. 

The precoding vector of mRZF at the $m$th AP for the users in the $n$th cluster can be expressed as
\begin{equation}\label{Eqn:precodeRZF}
\mathbf{w}_{mn} = \frac{\left(\bar{\mathbf{H}}_m \bar{\mathbf{H}}_m^\mathrm{H} + L\alpha \mathbf{I}_L\right)^{-1}\bar{\mathbf{h}}_{mn}}{\sqrt{\mathbb{E} \left\{ \left\lVert  (\bar{\mathbf{H}}_m \bar{\mathbf{H}}_m^\mathrm{H} + L\alpha \mathbf{I}_L)^{-1}\bar{\mathbf{h}}_{mn} \right \rVert^2 \right\}}},
\end{equation}
where $\bar{\mathbf{H}}_m$ is given in \eqref{Eqn:h_bar} and $\alpha >0$ is the regularization parameter that can be optimized \cite{Hoydis2013}. Finding optimal value for $\alpha$, however, is outside the scope of this paper and is left for future work. Here, we assume that $\alpha$ is scaled with $L$ to ensure that it converges to a constant value as $L$ and $N$ tend to infinity. 

It can be seen from \eqref{Eqn:precodeRZF}, in our design the precoder of each AP only utilizes local CSI rather than global CSI knowledge shared between the APs. Hence, mRZF has the same front-hauling overhead as that of MRT and fpZF.
\begin{thm}
In the CF massive MIMO-NOMA system with mRZF precoder, $\gamma_{nk}$ in \eqref{Eqn:INTERFERENCE} , when $L$ and $N$ grow large such that $ 1\le \text{lim}_{L,N\to \infty} \frac{L}{N} \le \infty$, , is given by
\eqref{Eqn:RZF_rate} as shown at the top of the next page, where $\mathbf{\Theta}_{mn} = {(1 +\tau p_p \sum\limits_{k=1}^k{\beta_{mnk}})}\mathbf{I}_L $, $a_{mnk}= \sqrt{\frac{\beta_{mnk}-\theta_{mnk}}{1+\tau p_p \sum_{k=1}^K{\beta_{mnk}}}}$ and $e_{mn}^o = e_{mn}$ in which
\begin{figure*}
\begin{equation}\label{Eqn:RZF_rate}
 \gamma_{nk}^{{\rm{mRZF}}} = \frac{ p_{nk}\left(\sum\limits_{m= 1}^M{  \frac{1}{\sqrt{\psi_{mn}^o}}\frac{c_{mnk} e_{mn}^o}{1 + e_{mn}^o}} \right)^2 }{ \left(\sum\limits_{m= 1}^M{  \frac{1}{\sqrt{\psi_{mn}^o}}\frac{c_{mnk} e_{mn}^o}{1 + e_{mn}^o}} \right)^2  \left(\sum\limits_{i=1}^{k-1}{p_{ni}} + \sum\limits_{i=k+1}^K{p_{ni}(2-2 \rho_{ni})}\right)+\sum\limits_{m = 1}^M{ \Upsilon_{mn}\left(\frac{c_{mnk}^2}{(1+e_{mn}^o)^2}+ a_{mnk}^2\right)}+1}.
\end{equation}
\rule{\textwidth}{0.25mm}
\end{figure*}

\begin{align}\label{TT1}
e_{mn} = \frac{1}{L} {\rm{tr}} \left[  \mathbf{\Theta}_{mn} \mathbf{T}_{m} \right] 
\end{align}
\begin{align}\label{TT2}
\mathbf{T}_{m} = \left( \frac{1}{L}\sum\limits_{j=1}^N {\frac{\mathbf{\Theta}_{mj}}{1+ e_{mj}}}+ \alpha \mathbf{I}_L\right)^{-1}
\end{align}
\begin{align}\label{TT3}
\psi_{mn}^o = \frac{1}{L} {\frac{ e'_{mn}}{(1+e_{mn})^2 }}
\end{align}
\begin{align}\label{TT4}
\Upsilon_{mn} = \frac{1}{L}\sum\limits_{\mystack{n'=1}{n' \ne 
n}}^N{\frac{p_{n'}e'_{n',mn}}{{\psi_{mn'}^o(1+e_{mn'})^2}}} 
\end{align}
in which $ \mathbf{e}'_m = [e'_{m1}, \ldots,e'_{mN}]^\mathrm{T}$ and  $\mathbf{e}'_{mn} = [e'_{1,mn}, \ldots,e'_{N,mn} ]^\mathrm{T}
$ are given by
\begin{align}
\mathbf{e}'_m &= \left( \mathbf{I}_N - \mathbf{J}_m\right)^{-1} \mathbf{v}_m,\\
\mathbf{e}'_{mn} &= \left( \mathbf{I}_N - \mathbf{J}_m\right)^{-1} \mathbf{v}_{mn},
\end{align}
and $\mathbf{J}_m$, $\mathbf{v}_m$ and $\mathbf{v}_{mn}$ are derived as follows,
\small{
\begin{align}\label{Eqn:ej1}
&[\mathbf{J}_m]_{ij} = \frac{\frac{1}{L}{\rm{tr}}\left[ \mathbf{\Theta}_{mi}\mathbf{T}_m\mathbf{\Theta}_{mj}\mathbf{T}_m\right]}{L(1+e_{mj})^2}\\
 &\mathbf{v}_m= \left[\frac{1}{L}{\rm{tr}}\left[  \mathbf{\Theta}_{m1}\mathbf{T}_m^2\right], \ldots, \frac{1}{L}{\rm{tr}} \left[ \mathbf{\Theta}_{mN}\mathbf{T}_m^2\right]  \right]^\mathrm{T}.\\
&\mathbf{v}_{mn} = \left[\frac{1}{L}{\rm{tr}}\left[ \mathbf{\Theta}_{m1}\mathbf{T}_m\mathbf{\Theta}_{mn}\mathbf{T}_m\right], \ldots, \frac{1}{L}{\rm{tr}}\left[ \mathbf{\Theta}_{mN}\mathbf{T}_m\mathbf{\Theta}_{mn}\mathbf{T}_m\right] \right]^\mathrm{T}.
\end{align}}
Besides, the initial values of $e_{mn}; \forall n$ to calculate \eqref{TT1} and \eqref{TT2} are  set to $e_{mn}^0 = \frac{1}{\alpha}, n=1,\ldots,N$ \cite{Wagner2012}.
\end{thm}

\begin{proof}
See Appendix \ref{RZF_app} for derivations.
\end{proof}
 \begin{rem}
In order to reduce the amount of overhead exchanged over the fronthaul network, channel estimation is performed locally at each AP and MRT,  fpZF  and mRZF precoders are designed based on the local CSI at each AP \eqref{Eqn:h_bar} . In this case, pilot signals are not shared over the fronthaul link. Indeed, only channel statistics (large scale parameters which changes slowly \cite{Ngo2017}) are sent to the CPU for user ordering and power allocation process. Thus, assuming a given system parameters $M$, $L$, $N$, and $K$, all three precoders have the same front-hauling overhead. In particular, for each realization of the user locations, the number of $MNK$ statistical parameters need to be exchanged between the CPU and the APs.
 \end{rem}
\section{Simulation Results}\label{sim}
Herein, we  provide  simulation results  to  evaluate the performance of the CF massive MIMO with NOMA or OMA. In NOMA, the same pilot sequence is shared among users within each cluster, but    different clusters are  assigned mutually orthogonal pilots. In  OMA, all  the  users  are  assigned  with  different  mutually  orthogonal  pilots. Therefore, the minimum pilot sequence lengths for NOMA and OMA are $\tau_{\text{NOMA}} = N$ and $\tau_{\text{OMA}} = KN$, respectively. The pre-log factor for NOMA and OMA cases are also defined as $\zeta_{\text{NOMA}}=(\tau_c-N)/\tau_c$ and $\zeta_{\text{OMA}}= (\tau_c-KN)/\tau_c$, respectively. 

In our simulations, the APs are uniformly distributed within an area of size $D \times D \quad {\rm{m}}^2$. Furthermore, users are clustered based on their spatial locations and all the clusters are uniformly distributed at random in the given area (the users in the same cluster are also distributed uniformly at random around the center point of the cluster).
\subsection{Large-Scale Fading  Model}
Here, we assume that the large-scale fading coefficient $\beta_{mnk}$ in \eqref{channel} includes both the  path-loss effect and shadowing. Thus, $\beta_{mnk}$ can be wrtten as \cite{Ngo2017}
\begin{equation}
\beta_{mnk} = {\rm{PL}}_{mnk} + z_{mnk}  \quad {\rm{dB}}, 
\end{equation}
where ${\rm{PL}}_{mnk}$ is the path loss,  $z_{mnk}$ represents the shadow fading generated from  a log-normal distribution with standard deviation $\sigma_{{\rm{sh}}}$. A three-slope model is considered for the path loss \cite{Tang2001}, the path loss exponent equals (i) $3.5$
if the distance between the $m$th AP and the $k$th user in the $n$th cluster ($d_{mnk}$) is greater than $d_1$, (ii) equals $2$ if $ d_0 <d_{mnk} \le d_1$, and (iii) equals $0$ if $d_{mnk} \le d_0$ for some $d_0$ and $d_1$. When $d_{mnk} > d_1$, the Hata-COSTA$231$ propagation model is employed \cite{Ngo2017}. Therefore, the path loss can be written as
{\small{\begin{equation}\label{pathloss}
 {\rm PL}_{mnk} =
\begin{cases}
 &-\mathcal{L}-35 {\rm{log}}_{10}(d_{mnk}),\qquad\qquad\quad\\
 &\qquad\qquad\qquad\qquad\qquad\qquad\quad\;\;\;\;d_{mnk} > d_1  \\ 
 &-\mathcal{L}-15 {\rm{log}}_{10}(d_1) - 20 {\rm{log}}_{10}(d_{mnk}),\\
&\qquad\qquad\qquad\qquad\qquad\quad\;\;\;\; d_0 <d_{mnk} \le d_1 \\ 
& -\mathcal{L}-15 {\rm{log}}_{10}(d_1)- 20 {\rm{log}}_{10}(d_{0}),\\
&\qquad\qquad\qquad\qquad\qquad\qquad\qquad\; d_{mnk} \le d_0 \\ 
\end{cases} 
\end{equation}}}
where
\begin{align}\label{constant}
\nonumber \mathcal{L} &\triangleq 46.3  +  33.9 {\rm{log}}_{10}(f) - 13.82{\rm{log}}_{10}(h_{AP}) \\&- (1.1 {\rm{log}}_{10}(f) - 0.7)h_u +(1.56 {\rm{log}}_{10}(f)-0.8).
\end{align}
In \eqref{constant}, $f$ is the carrier frequency (in MHz). Also, $h_{\rm AP}$ and $h_u$ are the AP antenna and user antenna heights (in m), respectively. Note  that when  $d_{mnk} \le d_1$, there is no shadowing.
\subsection{Parameters and Setup}
The simulation parameters are reported in Table I.  The noise variance is given by 
$\sigma^2_{w} = 290 \times k_b \times$ bandwidth $\times$ noise figure, where $k_b$ is the  Boltzmann constant. To each cluster, the $m$th AP  allocates   power  $\frac{p_{t}}{N}$. We further assume that $K =2$; i.e., each cluster has two users, and the total power allocated for the $n$th cluster $p_n$ is divided between the  two users within a cluster based on a $3:7$ ratio ($\forall n$) implying that $\lambda_{n1} = 0.3\lambda_{n}, \lambda_{n2} = 0.7\lambda_{n}$. However, these power coefficients are not necessarily optimal and can be  further optimized.   
\begin{table}\label{tab}
\caption{Simulation Settings} 
\centering 
\begin{tabular}{c c c c} 
\hline\hline 
Parameter & Value & Parameter & Value \\ [0.5ex]  
\hline 
Carrier frequency & 1.9 GHz & $\tau_c$ & 56 \\ 
Bandwidth & $20$ MHz & $p_{p}$ & 20 dBm\\
Noise figure & $9$ dB & $p_{t}$ & 23 dBm  \\
$D,d_1,d_0$ & $1000,50,10$ m & $\sigma_{sh}$ & $8$ dB \\
$h_{\rm{AP}}, h_{\rm{u}}$ & $65,15$ m  & $\rho_{ni}$ &$0.1$  \\ [1ex] 
\hline 
\end{tabular}
\label{table:nonlin} 
\end{table}
\subsection{Performance Evaluation}
Here, we first compare the performance of the CF massive MIMO-NOMA  and  OMA systems for  MRT, fpZF and mRZF precoding in terms of  sum rate given by \cite{Hanif2016} 
\begin{equation}\label{Eqn:achievble sum rate}
\mathcal{R} = \sum\limits_{n=1}^{N}{\sum\limits_{k=1}^K{\mathcal{R}_{nk}}}.
\end{equation}

\begin{figure}
\centering
\includegraphics[width=8.8cm, height=7.5cm]{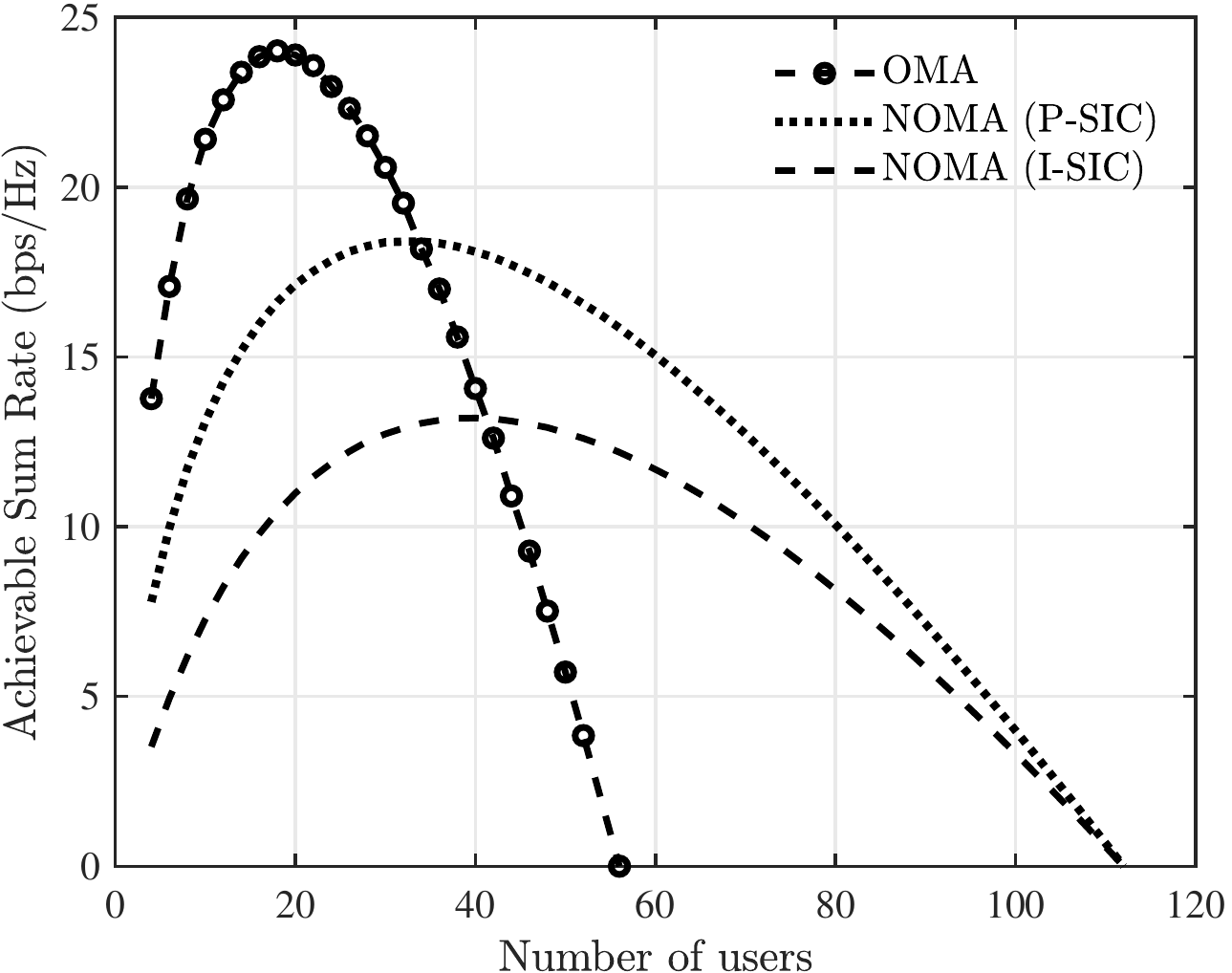}
\caption{The achievable sum rate versus the number of users for 25 APs (M = 25) and eight antennas per AP (L=8 with MRT precoding). P-SIC and I-SIC stand for  perfect SIC and imperfect SIC.}
\label{mrt_noma_oma} 
\end{figure}

Fig.~\ref{mrt_noma_oma} demonstrates the achievable sum rate \eqref{Eqn:achievble sum rate} of NOMA and OMA systems with MRT precoding as a function of the  number of users.  We set  $\tau_c = 56$.  We observe  that the maximum number of users  simultaneously served in OMA in the same resource block is $K_{\text{max}}^{{\rm{OMA}}} = 56$. In contrast,  NOMA doubles this. This is due to the fact that in OMA each user is assigned with an orthogonal pilot while in NOMA same pilot sequence is shared among users within each cluster. For a large number of users, NOMA outperforms OMA; however, with fewer  users, NOMA achieves  a  lower sum rate  than OMA due to the effects of intra-cluster pilot contamination and imperfect SIC. We also observe that, SIC imperfection  considerably degrades the performance of NOMA; for instance, compared to the perfect SIC, the residual interference caused by imperfect SIC degrades the achievable sum rate by $4.9 $ bps/Hz for $40$ simultaneously served users. 

In Fig.~\ref{fig:mrt_noma_oma_multiple_antenna}, the impact of deploying more antennas at the APs is investigated. As expected, adding more antennas at the APs  results in better sum rate, thanks to the array gain. However, despite OMA, the gain of adding more antennas diminishes for NOMA system since the interference due to the pilot contamination and imperfect SIC proportionally increases as $L$ increases. This observation confirms Remarks \ref{REM1_MTR} and \ref{remark}.

 \begin{figure}
 \centering
 \includegraphics[width=8.8cm, height=7.5cm]{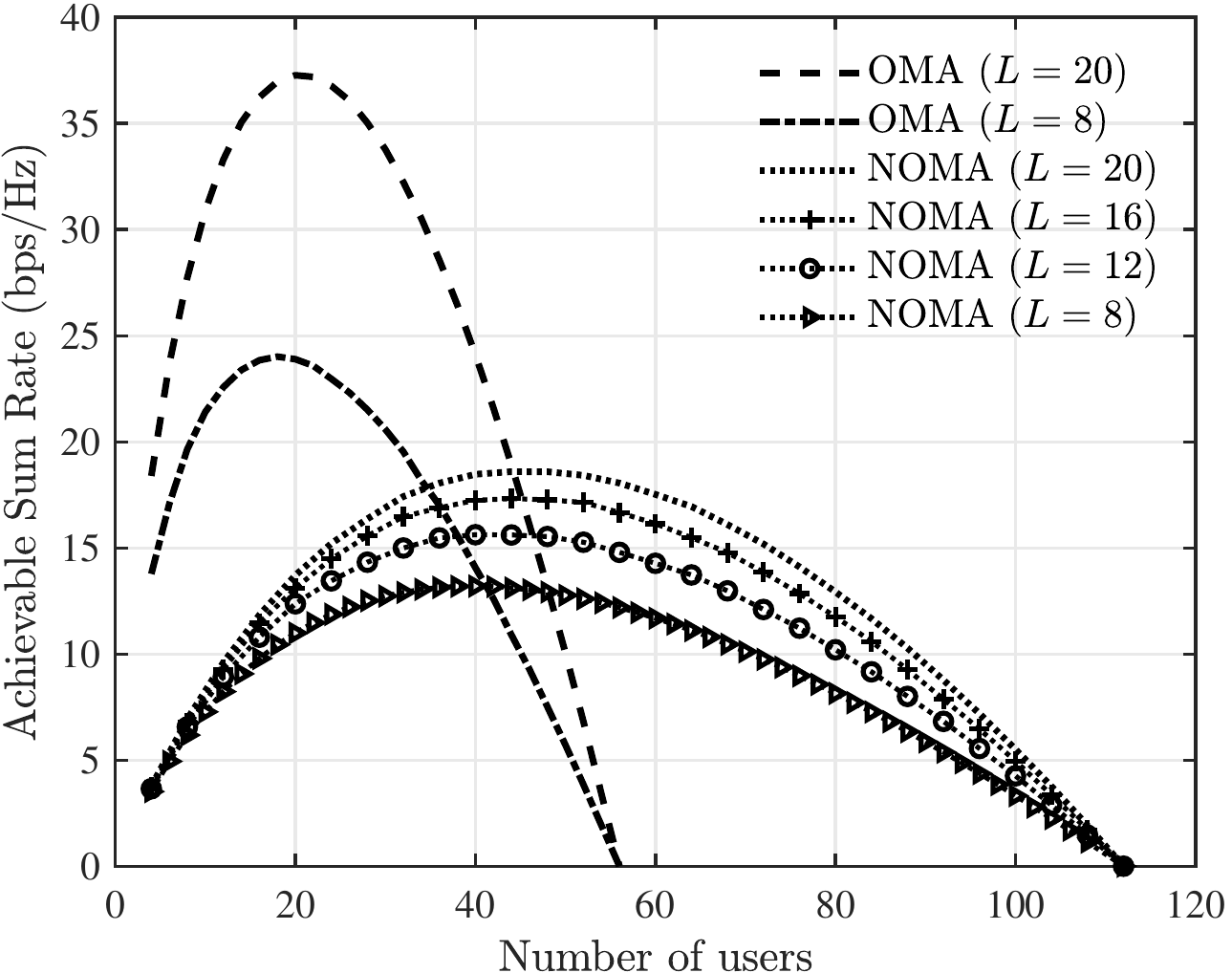}
\caption{The achievable sum rate versus the number of users for 25 APs  ($M = 25$, MRT, and  imperfect SIC).}
\label{fig:mrt_noma_oma_multiple_antenna} 
\end{figure}
Fig.~\ref{fig:ZF_noma_oma} shows the achievable sum rate of NOMA and OMA systems as a function of the  number of users. Here,  fpZF precoding and at  least  $ N+1\le L$ antennas \eqref{Eqn:rat2} per AP are considered. We see that  the residual interference caused by imperfect SIC significantly degrades the achievable sum rate. More specifically, for 40 users and 60 antennas per AP, the gap between I-SIC and P-SIC curves is about 53 bps/Hz.  For  imperfect SIC, the gain due  to antennas at APs is almost negligible.  This is because,  for a  low number of users (clusters), the SINR converges to \eqref{converge}, where imperfect  SIC  is the dominant term. Furthermore, as the number of users increases, both the desired power and interference power due to imperfect  SIC and pilot contamination increase proportionally to $L-N$. 
 \begin{figure}
 \centering
 \includegraphics[width=8.8cm, height=7.5cm]{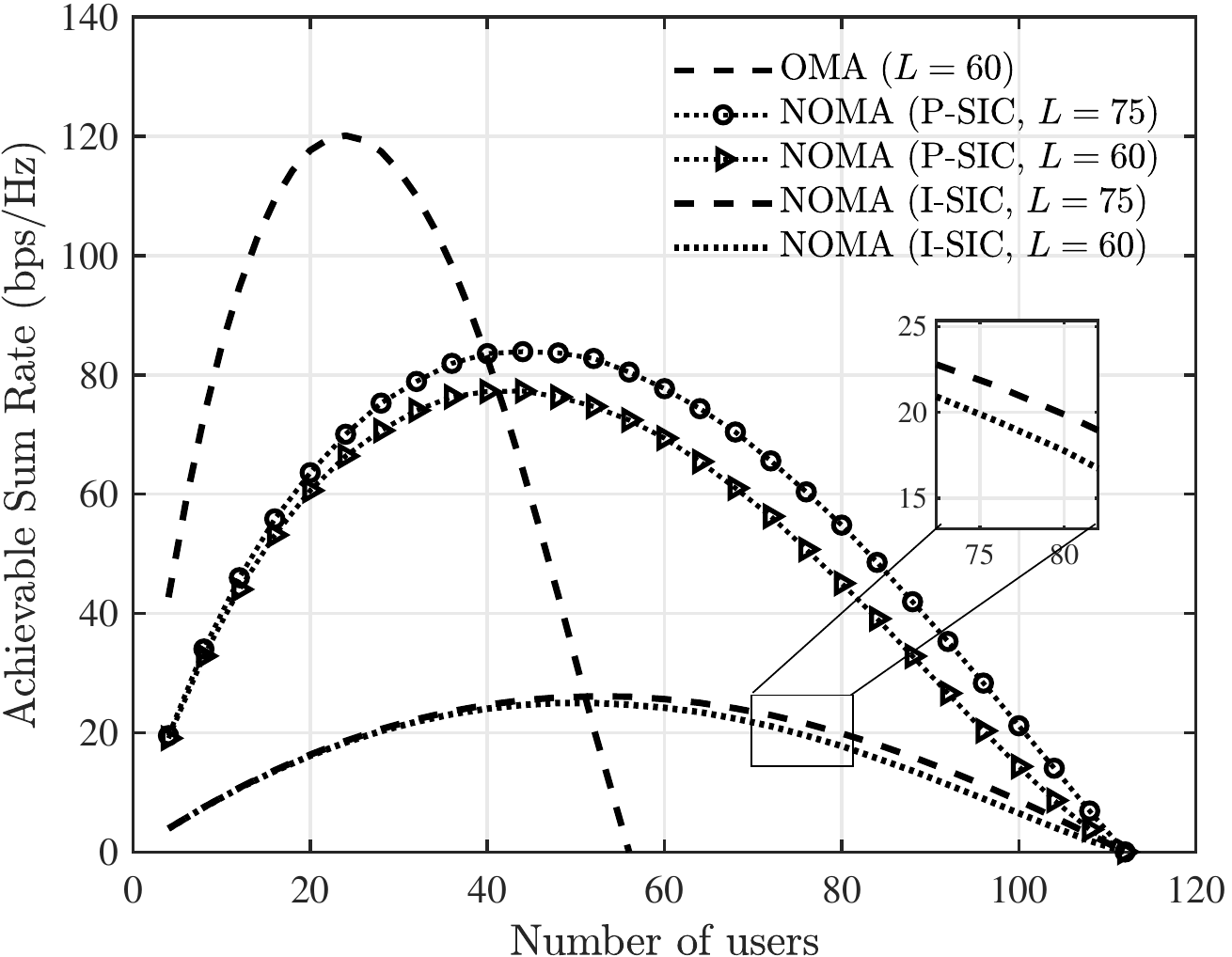}
\caption{The achievable sum rate  as a function of  the number of users for  25 APs (fpZF).}
\label{fig:ZF_noma_oma} 
\end{figure}

\begin{figure}
 \centering
 \includegraphics[width=8.8cm, height=7.5cm]{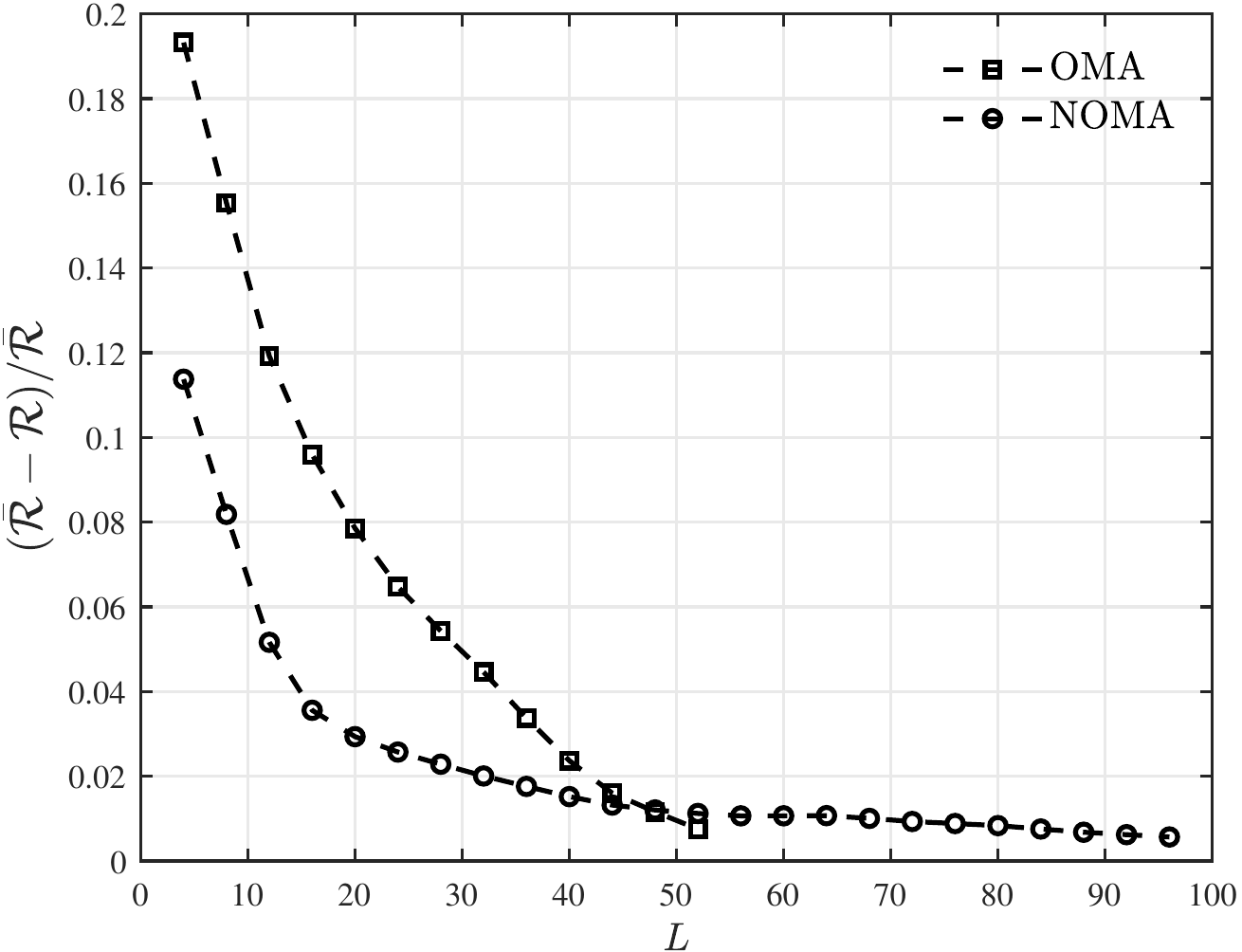}
\caption{The mRZF precoder; the relative error of the sum rate $\mathcal{R}$ via (29)  compared to the ergodic sum rate;  $(\bar{\mathcal{R}}-\mathcal{R})/\bar{\mathcal{R}}$ versus the number of antennas per AP,  $L = KN$,   for  25  APs and $\alpha = 0.8$.}
\label{fig:rzf_error} 
\end{figure}

For the mRZF precoder,  the SINR asymptotic  $\gamma_{nk}^{\text{mRZF}}$  in \eqref{Eqn:RZF_rate} must be validated. Thus,  Fig.~\ref{fig:rzf_error} shows the error of the sum rate $\mathcal{R}$ computed based on \eqref{Eqn:RZF_rate}   compared to the ergodic sum rate via simulations of the SINR. Clearly,  the  SINR asymptotic  holds even for small values of $L$ and becomes more accurate as $L$ increases. 
\begin{figure}
 \centering
 \includegraphics[width=8.8cm, height=7.5cm]{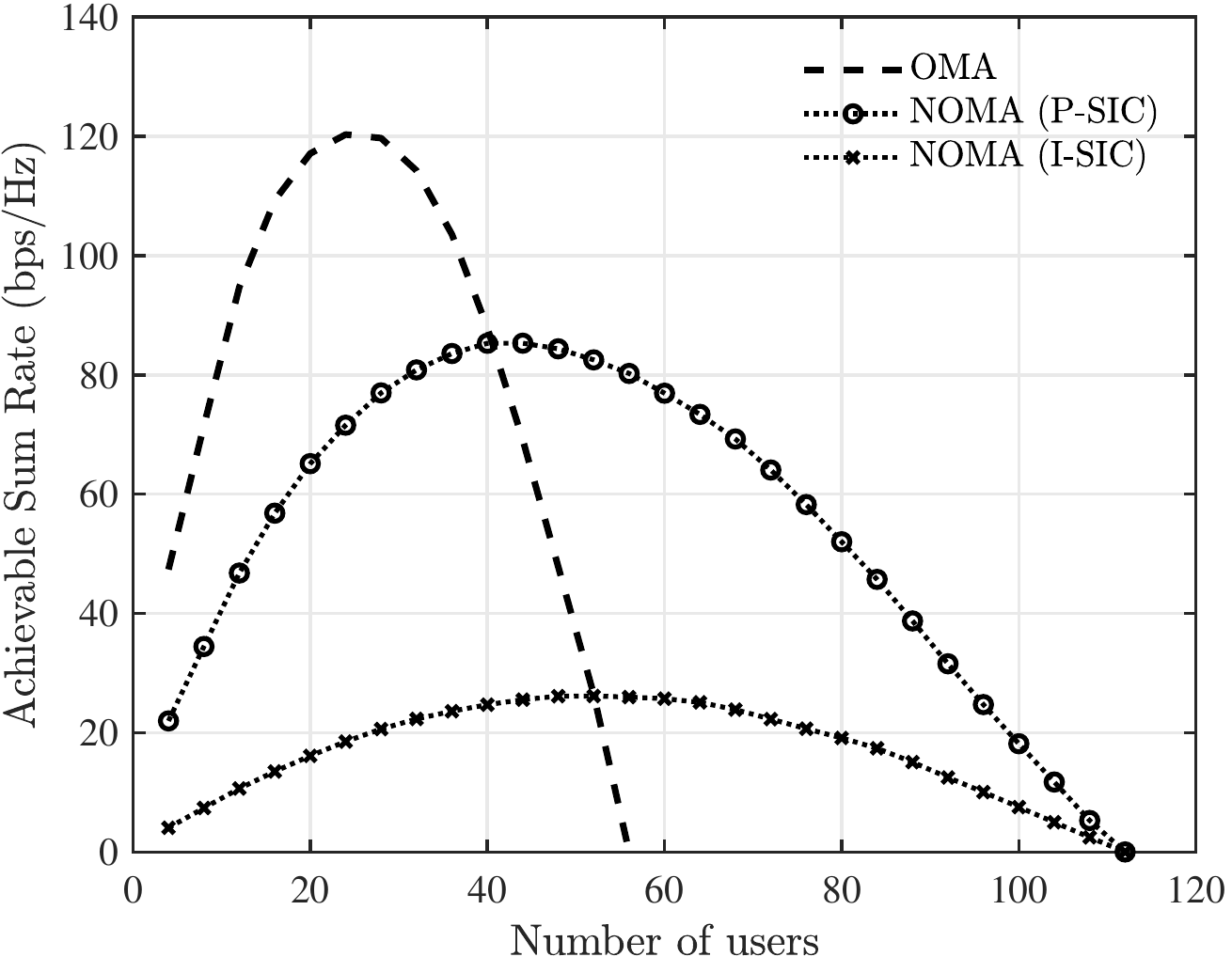}
\caption{The achievable sum rate versus the number of users for 25 APs and 60 antennas per AP and mRZF with  $\alpha = 0.8$.}
\label{fig:rzf_rate} 
\end{figure}

Fig.~\ref{fig:rzf_rate} shows the achievable sum rate of NOMA and OMA systems for different number of users with mRZF precoding and finite system dimensions. In this scheme, we simply consider $\alpha = 0.8$; however, the optimal  value of $ \alpha $ that accounts for the CSI imperfection could be further investigated \cite{Wagner2012, Zhu2016}. As expected, with   imperfect SIC, the resulting error propagation and    intra-cluster pilot contamination  limit the system performance. Just as in the case of  MRT and fpZF, for a large number of users, NOMA outperforms OMA; however, with a few users,  the converse is true.    

In Fig.~\ref{fig:compare_rate} and Fig.~\ref{fig:compare_rate_i}, we compare the performance of the three  precoders  for  perfect and imperfect SIC, respectively. We observe that, mRZF achieves higher rates than fpZF and MRT, since it tries  to  balance the inter-cluster interference mitigation and intra-cluster power enhancement. As well, for perfect SIC,
fpZF and mRZF outperform MRT as they are able to cancel the inter-cluster interference. Therefore, although fpZF and mRZF have the same front-hauling overhead as MRT, they can achieve higher rates. Furthermore, CF hybrid massive MIMO-NOMA  with either mRZF or fpZF outperforms  OMA systems with MRT as twice number of users could be served at the price of a negligible performance loss for a lightly loaded system.  We also observe from Fig.~\ref{fig:compare_rate_i}, with  imperfect SIC,  MRT and mRZF perform roughly the same for a  large number of users. This is because both fpZF and mRZF sacrifice  much  of the array gain to suppress the inter-cluster interference; however mRZF outperforms fpZF since it also considers the desired power. We also note that 
$\alpha$ may be optimized to enhance the performance of mRZF precoding.
\begin{figure}
 \centering
 \includegraphics[width=8.8cm, height=7.7cm]{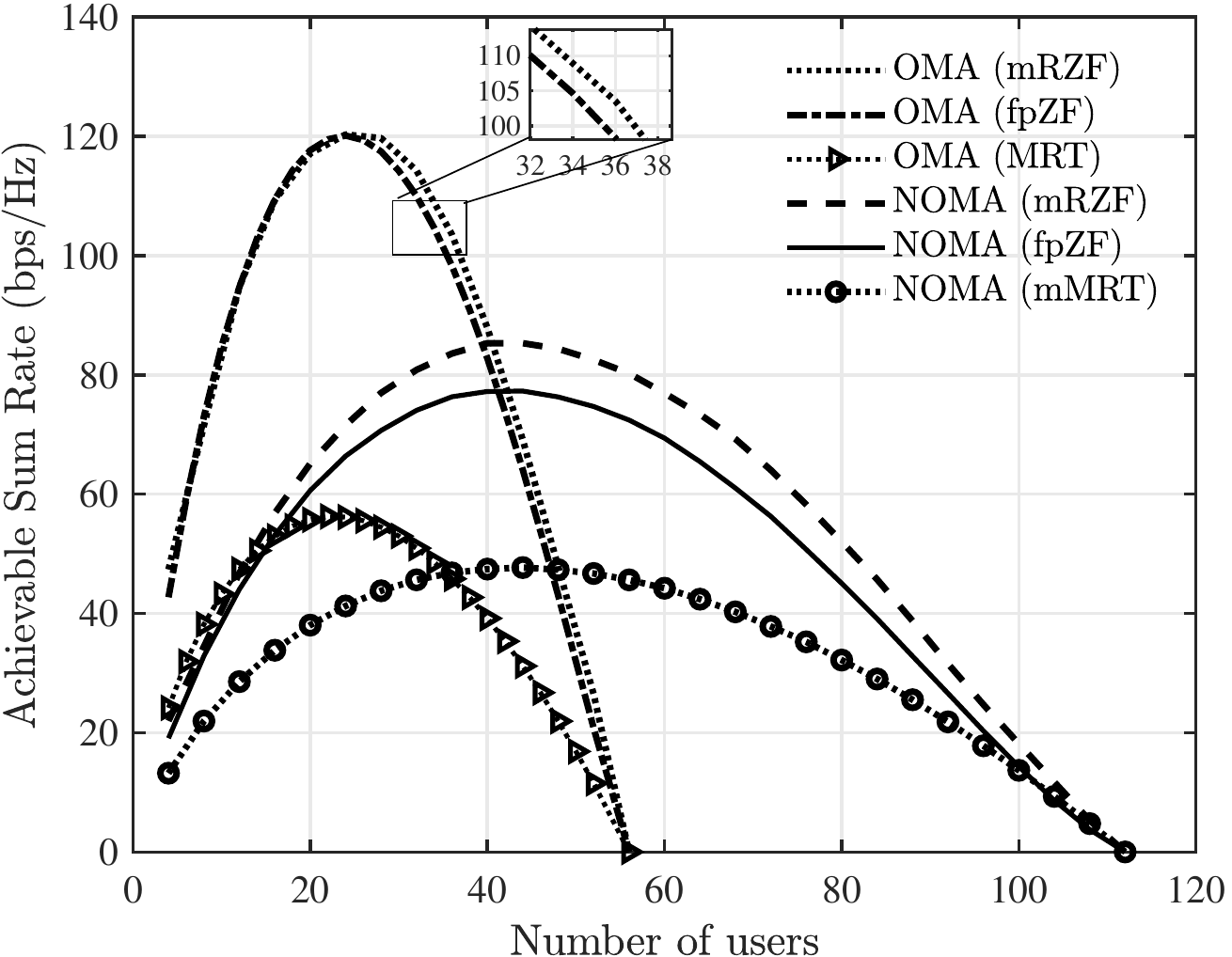}
\caption{The achievable sum rate versus the number of users for 25 APs and 60 antennas per AP and perfect SIC. }
\label{fig:compare_rate} 
\end{figure}

\begin{figure}
 \centering
 \includegraphics[width=8.8cm, height=7.5cm]{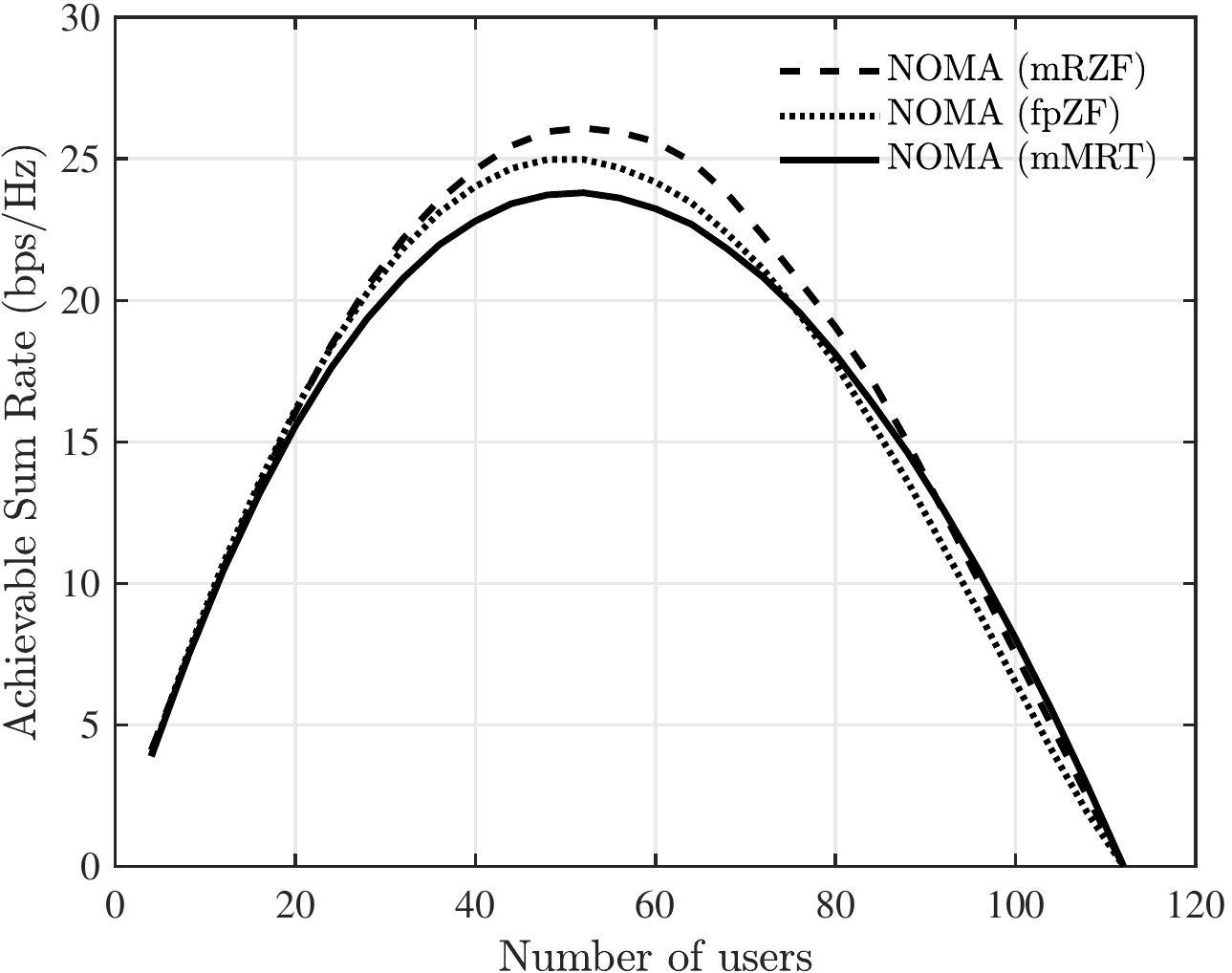}
\caption{The achievable sum rate versus the number of users for 25 APs and 60 antennas per AP (Imperfect SIC).}
\label{fig:compare_rate_i} 
\end{figure}
To further investigate these  three precoders, we also plot the cumulative distributions of their rates  per-cluster in Fig.~\ref{cdf}. It shows that  the per-cluster rate degrades due to  imperfect SIC. In particular, in $90 \%$-likely performance, all the three precoders behave almost the same and achieve low per-cluster rate of $\sim 1$ bps/Hz. In the case of perfect SIC, however, mRZF and fpZF significantly outperform  MRT. Thus, the rate loss due to the imperfect SIC is more considerable for fpZF and mRZF compared to that of MRT. More precisely, the $90 \%$-likely per-cluster rate loss is, respectively, about $1.8$ bps/Hz  and $2.2$ bps/Hz for fpZF and  mRZF, which are more than twice  than that of MRT ($0.8$ bps/Hz).
\begin{figure}
\centering
\includegraphics[width=8.8cm, height=7.5cm]{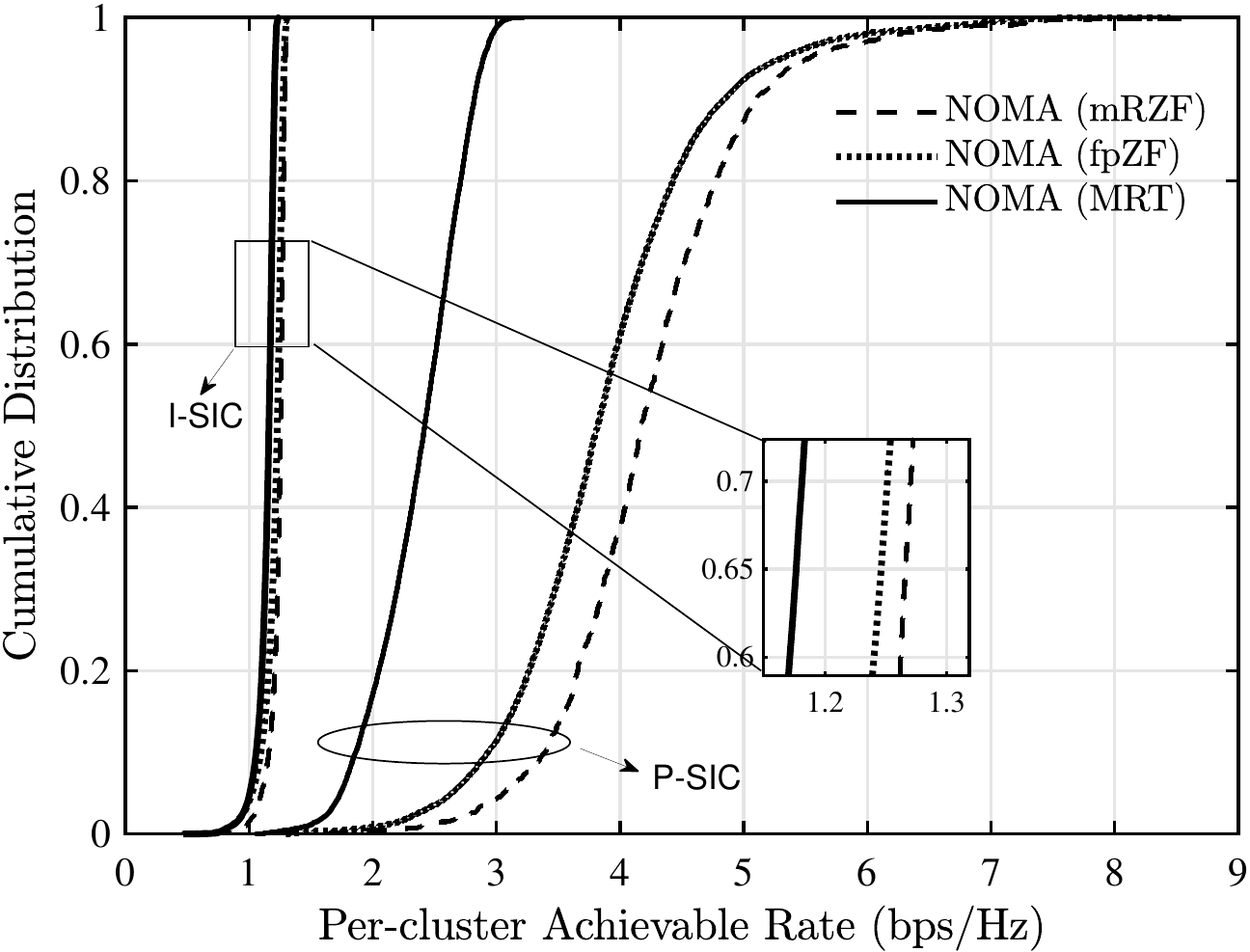}
\caption{Cumulative distribution of per-cluster achievable rate for  25 APs and 60 antennas per AP and  twenty clusters  (NOMA).}
\label{cdf} 
\end{figure}
\section{Conclusion}\label{conclusion}
In this paper, we analyzed the achievable rate of three standard precoders for   NOMA-aided CF massive MIMO.   We thus derived closed-form  sum rates for MRT and fpZF by   considering joint effects of intra-cluster pilot contamination, inter-cluster interference and imperfect SIC. However, since  a   closed-form  rate analysis is intractable for finite       mRZF precoder, we analyzed the  asymptotic regime, where the number of clusters and the number of antennas at each AP   grow extremely  large while keeping a finite ratio. 

We showed that  NOMA based CF massive MIMO  can support  significantly more  users compared to the OMA-hybrid  at the same time-frequency resources. For a large number of users, the NOMA-hybrid  outperforms the OMA-hybrid; however, for a few  users,  the former  achieves lower sum rate than  the latter  due to the effects of intra-cluster pilot contamination and imperfect SIC. It was further shown that, with  perfect SIC, mRZF and fpZF significantly outperforms MRT despite of having the same front-hauling overhead. Finally,  CF hybrid massive MIMO-NOMA  with either fpZF or mRZF  outperforms the  OMA systems with MRT. 

Future research includes the optimization of  $ \alpha $ of  mRZF precoder,  optimal user clustering and  power allocations. 
\vspace{-2ex}
\appendices
\section{Derivation of $\gamma_{nk}^{\rm MRT}$ \eqref{Eqn:mrt_rate}}\label{mrt_app}
  By invoking \eqref{Eqn:estimate}  in \eqref{Eqn:mrt}, the  MRT precoding vector   can be written as \begin{equation}\label{Eqn:mrt1}
\mathbf{w}_{mn} =  \frac{\hat{\mathbf{h}}_{mnk}}{\sqrt{L\theta_{mnk}}}.
\end{equation}

We first proceed to compute $\mathbb{E} \{\eta_{nk} \}$ as follows:
{\small{
\begin{align}\label{Eqn:nommrt}
\nonumber\mathbb{E} \{\eta_{nk} \} &= \mathbb{E} \left\{ \sum\limits_{m = 1}^M{ \mathbf{h}_{mnk}^\mathrm{H}\mathbf{w}_{mn}} \right\}\\
\nonumber &=\mathbb{E} \left\{\sum\limits_{m = 1}^M{(\hat{\mathbf{h}}_{mnk}^\mathrm{H} + \bm{\epsilon}_{mnk}^\mathrm{H}) \mathbf{w}_{mn}}\right\} \\
 &=  \sum\limits_{m = 1}^M{\sqrt{L\theta_{mnk}}}.
\end{align}}}
Since $ \mathbf{h}_{mnk}^\mathrm{H}\mathbf{w}_{mn}$ are independent for different values of $m$ and the variance of a sum of independent random variables is equal to the sum of the variances, we have
{\small{
\begin{align}\label{Eqn:r2mrt}
\nonumber& \mathbb{E} \left\{\left|(\eta_{nk} - \mathbb{E} \{\eta_{nk}\right\})\right|^2\} \\
 \nonumber &= \mathbb{E} \left\{\left|\sum\limits_{m = 1}^M{ \mathbf{h}_{mnk}^\mathrm{H}\mathbf{w}_{mn}}-  \mathbb{E} \left\{ \sum\limits_{m = 1}^M{ \mathbf{h}_{mnk}^\mathrm{H}\mathbf{w}_{mn}} \right\}\right|^2 \right\} \\
 \nonumber&= \sum\limits_{m = 1}^M{\mathbb{E} \left\{|\mathbf{h}_{mnk}^\mathrm{H}\mathbf{w}_{mn}|^2\right\}}-\sum\limits_{m = 1}^M{\left|\mathbb{E} \left\{\mathbf{h}_{mnk}^\mathrm{H}\mathbf{w}_{mn}\right\}\right|^2} \\
\nonumber &=  \sum\limits_{m = 1}^M{\mathbb{E} \left\{\left|(\hat{\mathbf{h}}_{mnk}^\mathrm{H}+ \bm{\epsilon}_{mnk}^\mathrm{H})\frac{\hat{\mathbf{h}}_{mnk}}{\sqrt{L\theta_{mnk}}}\right|^2\right\}}-\sum\limits_{m = 1}^M{L \theta_{mnk}}\\
 \nonumber &= \sum\limits_{m = 1}^M{\mathbb{E} \left\{\left|\hat{\mathbf{h}}_{mnk}^\mathrm{H}\frac{\hat{\mathbf{h}}_{mnk}}{\sqrt{L\theta_{mnk}}}\right|^2\right\}} + \sum\limits_{m = 1}^M{\mathbb{E} \left\{\left| \bm{\epsilon}_{mnk}^\mathrm{H}\frac{\hat{\mathbf{h}}_{mnk}}{\sqrt{L\theta_{mnk}}}\right|^2\right\}} \\ \nonumber &-\sum\limits_{m = 1}^M{L \theta_{mnk}} \\ 
\nonumber &\overset{(a)}= (L+1) \sum\limits_{m=1}^M{\theta_{mnk}} + \sum\limits_{m=1}^M{(\beta_{mnk}-\theta_{mnk})}- L\sum\limits_{m=1}^M{\theta_{mnk}} \\
&= \sum\limits_{m=1}^M{\beta_{mnk}},
\end{align}}}
where $(a)$ comes from the fact that, for any vector $\mathbf{z} \sim \mathcal{CN}(\mathbf{0},\theta \mathbf{I}_L)$, we have \cite{tulino2004}
\begin{equation}\label{Eqn:mean}
\mathbb{E} \left\{|\mathbf{z}^\mathrm{H} \mathbf{z}|^2 \right\} = (L^2+L)\theta^2.
\end{equation}
To find $\mathbb{E} \left\{|\eta_{nk}|^2\right\}$, we use   $
{\rm{Var}} \{ X\} = \mathbb{E} \{X^2\}- (\mathbb{E}\{X \})^2.$
Thus, we have
\begin{align}\label{Eqn:rmrt}
\nonumber \mathbb{E} \left\{|\eta_{nk}|^2\right\}&= \mathbb{E} \left\{\left|\left(\eta_{nk} - \mathbb{E} \left\{\eta_{nk}\right\}\right)\right|^2\right\} + (\mathbb{E} \{\eta_{nk}\})^2 \\
&= \sum\limits_{m=1}^M{\beta_{mnk}} + L \left(\sum\limits_{m=1}^M{\sqrt{\theta_{mnk}}}\right)^2.
\end{align}
Furthermore,
{{
\begin{align}\label{Eqn:r_primemrt}
\nonumber  &\mathbb{E} \left\{|\eta_{n'k}|^2\right\}= 
\mathbb{E} \left\{\left|\sum\limits_{m = 1}^M{ \mathbf{h}_{mnk}^\mathrm{H}\mathbf{w}_{mn'}}\right|^2 \right\}  \\
 \nonumber &=  \sum\limits_{m = 1}^M{\mathbb{E} \left\{\left|\mathbf{h}_{mnk}^\mathrm{H}\mathbf{w}_{mn'}\right|^2 \right\}}
 \\  \nonumber&+
 \sum\limits_{m = 1}^M{\sum\limits_{\mystack{m'=1}{m' \ne m}}^M{\mathbb{E} \left\{\left|\mathbf{h}_{mnk}^\mathrm{H}\mathbf{w}_{mn'} \left(\mathbf{h}_{m'nk}^\mathrm{H}\mathbf{w}_{m'n'} \right)\right| \right\}}}\\
& = \sum\limits_{m = 1}^M{\beta_{mnk}},
\end{align}}}
where the second term in the second equality is discarded as $\mathbf{h}_{mnk}^\mathrm{H}\mathbf{w}_{mn'} (\forall m \in 1,\ldots, M)$ are independent of each other.

Moreover, since $\mathbb{E} \{s_{ni}^* \hat{s}_{ni} \} = \mathbb{E} \{s_{ni} \hat{s}_{ni}^* \} = \rho_{ni}$ , we have
\begin{align}\label{Eqn:r3}
\nonumber  \mathbb{E} \left\{\left| \eta_{nk} s_{ni} - \mathbb{E} \left\{\eta_{nk}\right\} \hat{s}_{ni}\right|^2 \right\} &= \mathbb{E} \left\{| \eta_{nk}|^2 \right\} \\
&+ (1-2\rho_{ni})  \mathbb{E}^2 \{\eta_{nk}\}.
\end{align}
Finally, by substituting the above formulations into \eqref{Eqn:INTERFERENCE}, $\gamma_{nk}^{\rm MRT}$ can be obtained as \eqref{Eqn:mrt_rate}.

\section{Derivation of $\gamma_{nk}^{\rm ZF}$ \eqref{Eqn:rat2}}\label{zf_app}
For fpZF precoding, we first compute $\mathbb{E} \{\eta_{nk} \}$ as follows:
\begin{align}\label{Eqn:nom}
\nonumber\mathbb{E} \{\eta_{nk} \} &= \mathbb{E} \left\{ \sum\limits_{m = 1}^M{ \mathbf{h}_{mnk}^\mathrm{H}\mathbf{w}_{mn}} \right\} \\
\nonumber &=\mathbb{E} \left\{\sum\limits_{m = 1}^M{\left(\hat{\mathbf{h}}_{mnk}^\mathrm{H} + \bm{\epsilon}_{mnk}^\mathrm{H}\right) \mathbf{w}_{mn}}\right\} \\
 &= \sum\limits_{m = 1}^M{\sqrt{(L-N)\theta_{mnk}}}.
\end{align}
In the denominator, we have
\begin{equation}\label{Eqn:r2}
 \mathbb{E} \left\{\left|\left(\eta_{nk} - \mathbb{E} \{\eta_{nk}\}\right)\right|^2\right\} = \mathbb{E} \left\{|\eta_{nk}|^2\right\} - \mathbb{E}^2 \left\{\eta_{nk}\right\}.
\end{equation}
Referring to \eqref{Eqn:r3} and \eqref{Eqn:r2}, we just need to calculate $\mathbb{E} \left\{| \eta_{n'k}|^2 \right\}, \forall n'$. Employing the result in \eqref{Eqn:INTER}, we have
{{
\begin{align}\label{Eqn:r4}
 \nonumber &\mathbb{E} \left\{| \eta_{n'k}|^2 \right\} = \mathbb{E} \left\{\left|\sum\limits_{m = 1}^M{ \mathbf{h}_{mnk}^\mathrm{H}\mathbf{w}_{mn'}}\right|^2\right\}\\
\nonumber &= \mathbb{E} \left\{\left|\sum\limits_{m = 1}^M{ \left(\hat{\mathbf{h}}_{mnk}^\mathrm{H} +\bm{\epsilon}_{mnk}^\mathrm{H}\right) \mathbf{w}_{mn'}}\right|^2\right\}\\ 
    &= \left\{ \begin{array}{l}
(L-N) \left(\sum\limits_{m=1}^M {\sqrt{\theta_{mnk}}}\right)^2 + \sum\limits_{m=1}^M {(\beta_{mnk}- \theta_{mnk})},\\
\qquad\qquad\qquad\qquad\qquad\qquad\qquad\qquad\qquad\quad n = n'\\ 
\sum\limits_{m=1}^M {(\beta_{mnk}- \theta_{mnk})},  \!\!\!\!\!\!\!\!\! \IEEEeqnarraynumspace\IEEEeqnarraynumspace \IEEEeqnarraynumspace \IEEEeqnarraynumspace \IEEEeqnarraynumspace \IEEEeqnarraynumspace \;\;\;\;   n \neq n'  \\ 
 \end{array} \right .
 \end{align}}}
 
By substituting \eqref{Eqn:nom} and  \eqref{Eqn:r4} into \eqref{Eqn:INTERFERENCE},  $\gamma_{nk}^{{\rm{fpZF}}}$ is in  \eqref{Eqn:rat2}. 

\section{Derivation of $\gamma_{nk}^{\rm mRZF}$ \eqref{Eqn:RZF_rate}}\label{RZF_app}
To  derive SINR $\gamma_{nk}^{\rm mRZF}$, we need the following terms:
\begin{itemize}
    \item Desired signal power
      \begin{equation}\label{Eqn:desired power}
p_d = p_{nk}\left(\sum\limits_{m= 1}^M{  \mathbb{E} \left\{\mathbf{h}_{mnk}^\mathrm{H}\mathbf{w}_{mn} \right\}} \right)^2.
\end{equation}
 \item Variance of the beamforming gain uncertainty
 \small{
 \begin{align}\label{Eqn:BfGainUncertainty}
p_{_{I1}} =  p_{nk}\left(\sum\limits_{m = 1}^M{\mathbb{E} \left\{|\mathbf{h}_{mnk}^\mathrm{H}\mathbf{w}_{mn}|^2\right\}}-\sum\limits_{m = 1}^M{\left|\mathbb{E} \left\{\mathbf{h}_{mnk}^\mathrm{H}\mathbf{w}_{mn}\right\}\right|^2}\right). 
\end{align}}
\item Inter-cluster interference
\begin{align}\label{inter-cluster}
 p_{_{I2}} &= \sum\limits_{\mystack{n'=1}{n' \ne 
n}}^N{ p_{n'} \mathbb{E} \left\{ \left|\sum\limits_{m = 1}^M{ \mathbf{h}_{mnk}^\mathrm{H}\mathbf{w}_{mn'}}\right|^2 \right\}}.
\end{align}
\end{itemize}
In  order to calculate the intra-cluster interference due to the pilot contamination and imperfect SIC, we need 
$\mathbb{E} \{ |\sum_{m = 1}^M{ \mathbf{h}_{mnk}^\mathrm{H}\mathbf{w}_{mn}}|^2\}$, which can be derived by exploiting \eqref{Eqn:BfGainUncertainty} and the definition of the variance.

We now proceed to calculate \eqref{Eqn:desired power}, \eqref{Eqn:BfGainUncertainty} and \eqref{inter-cluster}  when $L,N \to \infty$ while keeping a finite ratio. 

Let 
\begin{align}\label{h_bar_difinition}
 \bar{\mathbf{h}}_{mn} =\sqrt{L}\mathbf{\Theta}_{mn}^{1/2} \bar{\mathbf{g}}_{mn}, 
\end{align}
where $\bar{\mathbf{g}}_{mn}$ has i.i.d complex entries  with zero mean and variance of $\frac{1}{L}$ (
$\bar{\mathbf{g}}_{mn}\sim \mathcal{CN}(\mathbf{0},\frac{1}{L}\mathbf{I}_L)$) and $\mathbf{\Theta}_{mn} = {(1 +\tau p_p \sum\limits_{k=1}^k{\beta_{mnk}})}\mathbf{I}_L $. Also, assume that $\mathbf{\Theta}_{mn};~ \forall m,n$ and $\frac{1}{L}\bar{\mathbf{H}}_m \bar{\mathbf{H}}_m^\mathrm{H} = \sum_{n=1}^{N}{\mathbf{\Theta}_{mn}^{1/2}\bar{\mathbf{g}}_{mn}\bar{\mathbf{g}}_{mn}^\mathrm{H}\mathbf{\Theta}_{mn}^{1/2}};~\forall m$ have uniformly bounded spectral norms \cite{Wagner2012}. Moreover, we define
\begin{align}\label{Eqn:definition1}
\mathbf{\Sigma}_m  \triangleq \left(\bar{\mathbf{H}}_m \bar{\mathbf{H}}_m^\mathrm{H} + L\alpha \mathbf{I}_L\right), \end{align}
\begin{align}\label{Eqn:definition2}
\mathbf{\Sigma}_{mn}  \triangleq \left(\bar{\mathbf{H}}_{mn} \bar{\mathbf{H}}_{mn}^\mathrm{H}  + L\alpha \mathbf{I}_L\right),
\end{align}
\begin{align}\label{Eqn:definition3}
\nonumber \psi_{mn}  &\triangleq {\mathbb{E} \left\{ \left\lVert  (\bar{\mathbf{H}}_m \bar{\mathbf{H}}_m^\mathrm{H} + L\alpha \mathbf{I}_L)^{-1}\bar{\mathbf{h}}_{mn} \right \rVert^2 \right\}}\\
&= {\mathbb{E} \left\{\bar{\mathbf{h}}_{mn}^\mathrm{H}\mathbf{\Sigma}_m^{-2} \bar{\mathbf{h}}_{mn}  \right\}}.
\end{align}
where $\bar{\mathbf{H}}_{mn}$ is equal to $\bar{\mathbf{H}}_{m}$ with the $n$th column removed.

\begin{itemize}
    \item Desired signal power
    
    Regarding \eqref{Eqn:desired power}, we derive $\mathbb{E} \{\mathbf{h}_{mnk}\mathbf{w}_{mn}\}$ as follows: 
{\small{
\begin{align}\label{Eqn:DesiredPower}
\nonumber &\mathbb{E} \left\{\mathbf{h}_{mnk}\mathbf{w}_{mn}\right\}  \\
\nonumber &=
\frac{1}{\sqrt{\psi_{mn}}}\mathbb{E} \left\{\mathbf{h}_{mnk}^\mathrm{H} \mathbf{\Sigma}_m^{-1} \bar{\mathbf{h}}_{mn}\right\}\\ 
\nonumber &\overset{(a)}= \frac{1}{\sqrt{\psi_{mn}}}\mathbb{E} \left\{\frac{\mathbf{h}_{mnk}^\mathrm{H} \mathbf{\Sigma}_{mn}^{-1} \bar{\mathbf{h}}_{mn}}{1 +\bar{\mathbf{h}}_{mn}^\mathrm{H}\mathbf{\Sigma}_{mn}^{-1} \bar{\mathbf{h}}_{mn}}\right\}\\
\nonumber &=\frac{1}{\sqrt{\psi_{mn}}} \mathbb{E} \left\{\frac{\left(c_{mnk}\bar{\mathbf{h}}_{mn} + \bm{\epsilon}_{mnk} \right)^\mathrm{H}\mathbf{\Sigma}_{mn}^{-1} \bar{\mathbf{h}}_{mn}}{1 +\bar{\mathbf{h}}_{mn}^\mathrm{H}\Sigma_{mn}^{-1} \bar{\mathbf{h}}_{mn}}\right\}\\
\nonumber &=\frac{1}{\sqrt{\psi_{mn}}} \mathbb{E} \left\{\frac{c_{mnk}\bar{\mathbf{h}}_{mn}^\mathrm{H} \mathbf{\Sigma}_{mn}^{-1} \bar{\mathbf{h}}_{mn}}{1 +\bar{\mathbf{h}}_{mn}^\mathrm{H}\mathbf{\Sigma}_{mn}^{-1} \bar{\mathbf{h}}_{mn}} + \frac{ \bm{\epsilon}_{mnk}^\mathrm{H}\mathbf{\Sigma}_{mn}^{-1} \bar{\mathbf{h}}_{mn}}{1 +\bar{\mathbf{h}}_{mn}^\mathrm{H}\mathbf{\Sigma}_{mn}^{-1} \bar{\mathbf{h}}_{mn}} \right\},\\
\end{align}}}
where $(a)$ comes from the matrix inversion Lemma \cite{Couillet,Wagner2012}. 
In order to solve \eqref{Eqn:DesiredPower}, we use the following Lemma.

\begin{lem} \label{lemma1}
Let $\mathbf{A}\in \mathcal{C}^{L \times L}$ and $\mathbf{x},\mathbf{y} \sim \mathcal{CN}(\mathbf{0}, \frac{1}{L}\mathbf{I}_L)$. Assume that $\mathbf{A}$ has uniformly bounded spectral norm (with respect to $L$) and that $\mathbf{x}$ and $\mathbf{y}$ are mutually independent and independent of $\mathbf{A}$. Now, as ${L} \longrightarrow \infty$,
\begin{align}
 a.~\mathbf{x}^\mathrm{H} \mathbf{A} \mathbf{x} \mathop \to \limits^{a.s.} \frac{1}{L} {\rm{tr}}\left[\mathbf{A}\right] , ~ 
 b.~\mathbf{x}^\mathrm{H} \mathbf{A} \mathbf{y} \mathop \to \limits^{a.s.}0,
\end{align}
\end{lem}
where $\mathop  \to \limits^{a.s.}$ denotes almost sure convergence as $L \longrightarrow \infty$. 

By substituting $\bar{\mathbf{h}}_{mn}$ with \eqref{h_bar_difinition} and applying  Lemma \ref{lemma1}.a and Theorem 2 of \cite{Wagner2012} to the first term of \eqref{Eqn:DesiredPower}, we have
\begin{align}\label{Eqn:desireda}
\nonumber \frac{c_{mnk}\bar{\mathbf{h}}_{mn}^\mathrm{H} \mathbf{\Sigma}_{mn}^{-1} \bar{\mathbf{h}}_{mn}}{1 +\bar{\mathbf{h}}_{mn}^\mathrm{H}\mathbf{\Sigma}_{mn}^{-1} \bar{\mathbf{h}}_{mn}} &= \frac{c_{mnk}L   \bar{\mathbf{g}}_{mn}^\mathrm{H} \mathbf{\Theta}_{mn}^{1/2} \mathbf{\Sigma}_{mn}^{-1} \mathbf{\Theta}_{mn}^{1/2} \bar{\mathbf{g}}_{mn}}{1 +L   \bar{\mathbf{g}}_{mn}^\mathrm{H} \mathbf{\Theta}_{mn}^{1/2} \mathbf{\Sigma}_{mn}^{-1} \mathbf{\Theta}_{mn}^{1/2} \bar{\mathbf{g}}_{mn}} \\
\nonumber &= \frac{c_{mnk} \bar{\mathbf{g}}_{mn}^\mathrm{H} \mathbf{\Theta}_{mn}^{1/2} \mathbf{C}_{mn}^{-1} \mathbf{\Theta}_{mn}^{1/2} \bar{\mathbf{g}}_{mn}}{1 +  \bar{\mathbf{g}}_{mn}^\mathrm{H} \mathbf{\Theta}_{mn}^{1/2} \mathbf{C}_{mn}^{-1} \mathbf{\Theta}_{mn}^{1/2} \bar{\mathbf{g}}_{mn}} \\
&\mathop \to \limits^{a.s.} \frac{c_{mnk} e_{mn}^o}{1 +e_{mn}^o},
\end{align}
where $\mathbf{C}_{mn} = \mathbf{\Gamma}_{mn} + \alpha \mathbf{I}_L$ with  $\mathbf{\Gamma}_{mn} = \frac{1}{L}\bar{\mathbf{H}}_{mn} \bar{\mathbf{H}}_{mn}^\mathrm{H} $. Besides, $e_{mn}^o$ is given in \eqref{TT1}.
Similarly, Since, $\hat{\mathbf{h}}_{mnk} = c_{mnk} \bar{\mathbf{h}}_{mn}$ is independent of $\bm{\epsilon}_{mnk}$, by applying Lemma \ref{lemma1}.b to the second term of \eqref{Eqn:DesiredPower}, we have
\begin{align}\label{Eqn:error}
\bm{\epsilon}_{mnk}^\mathrm{H}\mathbf{\Sigma}_{mn}^{-1} \bar{\mathbf{h}}_{mn}\mathop \to \limits^{a.s.} 0.
\end{align}
Regarding \eqref{Eqn:definition3}, to find the value of $\psi_{mn}$, we need to calculate
$\mathbb{E} \left\{\bar{\mathbf{h}}_{mn}^\mathrm{H} \mathbf{\Sigma}_m^{-2} \bar{\mathbf{h}}_{mn}\right\}$.

By employing matrix inversion Lemma, Theorems 1 and  2 of \cite{Wagner2012}, we obtain
\begin{align}\label{Eqn:gamma1}
\nonumber{ {\bar{\mathbf{h}}_{mn}^\mathrm{H} \mathbf{\Sigma}_m ^{-2}\bar{\mathbf{h}}_{mn}}}
 &\overset{(a)}= \frac{1}{L} {\frac{\bar{\mathbf{g}}_{mn}^\mathrm{H}\mathbf{\Theta}_{mn}^{1/2} \mathbf{C}_{mn}^{-2}\mathbf{\Theta}_{mn}^{1/2} \bar{\mathbf{g}}_{mn}}{(1+\bar{\mathbf{g}}_{mn}^\mathrm{H} \mathbf{\Theta}_{mn}^{1/2}\mathbf{C}_{mn}^{-1} \mathbf{\Theta}_{mn}^{1/2}\bar{\mathbf{g}}_{mn})^2 }}\\
\nonumber &\mathop \to \limits^{a.s.}\frac{1}{L} {\frac{\frac{1}{L} {\rm{tr}}\left[ \mathbf{\Theta}_{mn}\mathbf{C}_{mn}^{-2}\right] }{(1+\frac{1}{L}{\rm{tr}} \left[ \mathbf{\Theta}_{mn}\mathbf{C}_{mn}^{-1}\right])^2 }} \\
&\mathop \to \limits^{a.s.} \frac{1}{L} {\frac{\frac{1}{L} {\rm{tr}} \left[\mathbf{\Theta}_{mn}\mathbf{T}'_{m}\right]  }{(1+\frac{1}{L}{\rm{tr}}\left[ \mathbf{\Theta}_{mn}\mathbf{T}_{m}\right])^2 }},
\end{align}
where $(a)$ comes from the matrix inversion Lemma which is applied twice, $\mathbf{T}_{mn}$ is defined in \eqref{TT2} and
\begin{align}\label{Eqn:tt2}
\mathbf{T}'_{m} = \mathbf{T}_{m} \left[\frac{1}{L} \sum\limits_{j =1}^{N}{\frac{ \mathbf{\Theta}_{mj} e'_{mj} }{(1+e_{mj})^2 }}+ \mathbf{I}_L\right]\mathbf{T}_{m},
\end{align}
where
\begin{align}\label{Eqn:tt3}
e'_{mj} = \frac{1}{L} {\rm{tr}} \left[ \mathbf{\Theta}_{mj}\mathbf{T}'_{m}\right].
\end{align}
By employing Theorem 2 of \cite{Wagner2012}, we finally have
\begin{align}\label{Eqn:gamma_fin}
 {\bar{\mathbf{h}}_{mn}^\mathrm{H} \mathbf{\Sigma}_m ^{-2}\bar{\mathbf{h}}_{mn}}\mathop \to \limits^{a.s.} \psi_{mn}^o,
\end{align}
where $\psi_{m}^o $ is given in \eqref{TT3}.
Therefore, by invoking \eqref{Eqn:desireda} and \eqref{Eqn:gamma_fin} in \eqref{Eqn:desired power}, desired signal power can be written as
\begin{align}\label{Eqn:desired_n}
p_d \mathop \to \limits^{a.s.} p_{n,k}\left(\sum\limits_{m= 1}^M{  \frac{1}{\sqrt{\psi_{mn}^o}}\frac{c_{mnk} e_{mn}^o}{1 + e_{mn}^o}} \right)^2.
\end{align}
\item Variance of the beamforming gain uncertainty\\
Regarding \eqref{Eqn:BfGainUncertainty}, We need to calculate the term ($\mathbb{E} \left\{|\mathbf{h}_{mnk}^\mathrm{H}\mathbf{w}_{mn}|^2\right\}$). By employing \eqref{Eqn:DesiredPower} and \eqref{Eqn:desireda}, we have
\begin{eqnarray}\label{Eqn:gain uncertainty2}
 |\mathbf{h}_{mnk}^\mathrm{H} \mathbf{w}_{mn}|^2 
\mathop \to \limits^{a.s.} \frac{1}{\psi_{mn}^o}\frac{c_{mnk}^2 (e_{mn}^o)^2}{(1 + e_{mn}^o)^2}.
\end{eqnarray}
Substituting \eqref{Eqn:desireda}, \eqref{Eqn:gamma_fin} and \eqref{Eqn:gain uncertainty2} in \eqref{Eqn:BfGainUncertainty}, we obtain
{\small{
\begin{align}\label{Eqn:gain uncertainty3}
 \nonumber p_{_{I1}} &\mathop \to \limits^{a.s.} p_{nk} \left(\sum\limits_{m = 1}^M{\frac{1}{\psi_{mn}^o}\frac{c_{mnk}^2 (e_{mn}^o)^2}{(1 + e_{mn}^o)^2}}-\sum\limits_{m = 1}^M{\frac{1}{\psi_{mn}^o}\frac{c_{mnk}^2 (e_{mn}^o)^2}{(1 + e_{mn}^o)^2}}\right)\\
 &\mathop \to \limits^{a.s.} 0.
\end{align}}}
We also find 
\begin{align}\label{SIC_INTER}
\mathbb{E} \left\{ \left|\sum\limits_{m = 1}^M{ \mathbf{h}_{mnk}^\mathrm{H}\mathbf{w}_{mn}}\right|^2 \right\} \mathop \to \limits^{a.s.}  \left(\sum\limits_{m= 1}^M{  \frac{1}{\sqrt{\psi_{mn}^o}}\frac{c_{mnk} e_{mn}^o}{1 + e_{mn}^o}} \right)^2. 
\end{align}

\item Inter-cluster interference

In order to derive \eqref{inter-cluster}, we define
\begin{equation}\label{pow}
\mathbf{P}_{mn} \triangleq {\rm{Diag}}\left(\frac{p_{1}}{\psi_{m1}^o},\ldots, \frac{p_{n-1}}{\psi_{m(n-1)}^o},\frac{p_{n+1}}{\psi_{m(n+1)}^o}, \ldots,\frac{p_{N}}{\psi_{mN}^o}\right).
\end{equation}
By applying  the  matrix inversion Lemma, Lemma \ref{lemma1} and \eqref{pow}, \eqref{inter-cluster} can be written as
\begin{align}\label{interferenc2}
\nonumber p_{_{I2}}&=\sum\limits_{\mystack{n'=1}{n' \ne 
n}}^N{p_{n'} \sum\limits_{m = 1}^M{\mathbb{E} \left\{ \left|\mathbf{h}_{mnk}^\mathrm{H}\mathbf{w}_{mn'}\right|^2 \right\}}} \\
 &=
\sum\limits_{m = 1}^M{\mathbb{E} \left\{ {\mathbf{h}}_{mnk}^\mathrm{H} \mathbf{\Sigma}_{m}^{-1} \bar{\mathbf{H}}_{mn} \mathbf{P}_{mn} \bar{\mathbf{H}}_{mn}^\mathrm{H}\mathbf{\Sigma}_{m}^{-1}{\mathbf{h}}_{mnk} \right \}}.
\end{align}
Since ${\mathbf{h}}_{mnk} = {\hat{\mathbf{h}}}_{mnk} + \bm{\epsilon}_{mnk}$ and $\hat{\mathbf{h}}_{mnk} = c_{mnk} \bar{\mathbf{h}}_{mn}$, we find that the estimated channels of the  users  in  the  same  cluster  are parallel. Accordingly, ${\mathbf{h}}_{mnk}$ can be written as
\begin{align} \label{h_mat}
 {\mathbf{h}}_{mnk} = \sqrt{L} \mathbf{\Theta}_{mn}^{1/2}[c_{mnk}\bar{\mathbf{g}}_{mn}+ a_{mnk} \hat{\mathbf{g}}_{mnk}],  
\end{align}
where $a_{mnk}= \sqrt{\frac{\beta_{mnk}-\theta_{mnk}}{1+\tau p_p \sum_{k=1}^K{\beta_{mnk}}}}$ and $\hat{\mathbf{g}}_{mnk}\sim \mathcal{CN}(\mathbf{0},\frac{1}{L}\mathbf{I}_L) $ is independent of $\bar{\mathbf{g}}_{mn}$. By substituting \eqref{h_mat} in \eqref{interferenc2}, we find
\small{
\begin{align}\label{deterministic}
\nonumber &{\mathbf{h}}_{mnk}^\mathrm{H} \mathbf{\Sigma}_{m}^{-1} \bar{\mathbf{H}}_{mn} \mathbf{P}_{mn} \bar{\mathbf{H}}_{mn}^\mathrm{H}\mathbf{\Sigma}_{m}^{-1}{\mathbf{h}}_{mnk} \\
\nonumber &= \frac{1}{L}c_{mnk}^2\bar{\mathbf{g}}_{mn}^\mathrm{H}\mathbf{\Theta}_{mn}^{1/2}\mathbf{C}_{m}^{-1}\bar{\mathbf{H}}_{mn}\mathbf{P}_{mn} \bar{\mathbf{H}}_{mn}^\mathrm{H}\mathbf{C}_{m}^{-1}\mathbf{\Theta}_{mn}^{1/2}\bar{\mathbf{g}}_{mn} \\
\nonumber &+ \frac{1}{L}a_{mnk}^2\hat{\mathbf{g}}_{mnk}^\mathrm{H}\mathbf{\Theta}_{mn}^{1/2}\mathbf{C}_{m}^{-1}\bar{\mathbf{H}}_{mn}\mathbf{P}_{mn} \bar{\mathbf{H}}_{mn}^\mathrm{H}\mathbf{C}_{m}^{-1}\mathbf{\Theta}_{mn}^{1/2}\hat{\mathbf{g}}_{mnk}\\
\nonumber &+ \frac{1}{L}c_{mnk}a_{mnk}\bar{\mathbf{g}}_{mn}^\mathrm{H}\mathbf{\Theta}_{mn}^{1/2}\mathbf{C}_{m}^{-1}\bar{\mathbf{H}}_{mn}\mathbf{P}_{mn} \bar{\mathbf{H}}_{mn}^\mathrm{H}\mathbf{C}_{m}^{-1}\mathbf{\Theta}_{mn}^{1/2}\hat{\mathbf{g}}_{mnk} \\
 &+\frac{1}{L}c_{mnk}a_{mnk}\hat{\mathbf{g}}_{mnk}^\mathrm{H}\mathbf{\Theta}_{mn}^{1/2}\mathbf{C}_{m}^{-1}\bar{\mathbf{H}}_{mn}\mathbf{P}_{mn} \bar{\mathbf{H}}_{mn}^\mathrm{H}\mathbf{C}_{m}^{-1}\mathbf{\Theta}_{mn}^{1/2}\bar{\mathbf{g}}_{mn}.
\end{align}}
where $\mathbf{C}_{m} = \mathbf{\Gamma}_{m} + \alpha \mathbf{I}_L$ with $\mathbf{\Gamma}_{m} = \frac{1}{L}\bar{\mathbf{H}}_{m} \bar{\mathbf{H}}_{m}^\mathrm{H} $. The equation in  \eqref{deterministic} can be further written as \eqref{determinis}.
\begin{figure*}
\small{
\begin{align}\label{determinis}
\nonumber {\mathbf{h}}_{mnk}^\mathrm{H} \mathbf{\Sigma}_{m} \bar{\mathbf{H}}_{mn} \mathbf{P}_{mn} \bar{\mathbf{H}}_{mn}^\mathrm{H}\mathbf{\Sigma}_{m}{\mathbf{h}}_{mnk}
 &= \frac{1}{L}c_{mnk}^2\bar{\mathbf{g}}_{mn}^\mathrm{H}\mathbf{\Theta}_{mn}^{1/2}\mathbf{C}_{mn}^{-1}\bar{\mathbf{H}}_{mn}\mathbf{P}_{mn} \bar{\mathbf{H}}_{mn}^\mathrm{H}\mathbf{C}_{m}^{-1}\mathbf{\Theta}_{mn}^{1/2}\bar{\mathbf{g}}_{mn} \\
\nonumber &+ \frac{1}{L}c_{mnk}^2\bar{\mathbf{g}}_{mn}^\mathrm{H}\mathbf{\Theta}_{mn}^{1/2}(\mathbf{C}_{m}^{-1}- \mathbf{C}_{mn}^{-1})\bar{\mathbf{H}}_{mn}\mathbf{P}_{mn} \bar{\mathbf{H}}_{mn}^\mathrm{H}\mathbf{C}_{m}^{-1}\mathbf{\Theta}_{mn}^{1/2}\bar{\mathbf{g}}_{mn}\\
\nonumber &+ \frac{1}{L}a_{mnk}^2\hat{\mathbf{g}}_{mnk}^\mathrm{H}\mathbf{\Theta}_{mn}^{1/2}\mathbf{C}_{mn}^{-1}\bar{\mathbf{H}}_{mn}\mathbf{P}_{mn} \bar{\mathbf{H}}_{mn}^\mathrm{H}\mathbf{C}_{m}^{-1}\mathbf{\Theta}_{mn}^{1/2}\hat{\mathbf{g}}_{mnk}\\
\nonumber &+ \frac{1}{L}a_{mnk}^2\hat{\mathbf{g}}_{mnk}^\mathrm{H}\mathbf{\Theta}_{mn}^{1/2}(\mathbf{C}_{m}^{-1}-\mathbf{C}_{mn}^{-1})\bar{\mathbf{H}}_{mn}\mathbf{P}_{mn} \bar{\mathbf{H}}_{mn}^\mathrm{H}\mathbf{C}_{m}^{-1}\mathbf{\Theta}_{mn}^{1/2}\hat{\mathbf{g}}_{mnk}\\
\nonumber &+ \frac{1}{L}c_{mnk}a_{mnk}\bar{\mathbf{g}}_{mn}^\mathrm{H}\mathbf{\Theta}_{mn}^{1/2}\mathbf{C}_{mn}^{-1}\bar{\mathbf{H}}_{mn}\mathbf{P}_{mn} \bar{\mathbf{H}}_{mn}^\mathrm{H}\mathbf{C}_{m}^{-1}\mathbf{\Theta}_{mn}^{1/2}\hat{\mathbf{g}}_{mnk} \\
\nonumber &+ \frac{1}{L}c_{mnk}a_{mnk}\bar{\mathbf{g}}_{mn}^\mathrm{H}\mathbf{\Theta}_{mn}^{1/2}(\mathbf{C}_{m}^{-1}-\mathbf{C}_{mn}^{-1})\bar{\mathbf{H}}_{mn}\mathbf{P}_{mn} \bar{\mathbf{H}}_{mn}^\mathrm{H}\mathbf{C}_{m}^{-1}\mathbf{\Theta}_{mn}^{1/2}\hat{\mathbf{g}}_{mnk}\\
\nonumber &+\frac{1}{L}c_{mnk}a_{mnk}\hat{\mathbf{g}}_{mnk}^\mathrm{H}\mathbf{\Theta}_{mn}^{1/2}\mathbf{C}_{mn}^{-1}\bar{\mathbf{H}}_{mn}\mathbf{P}_{mn} \bar{\mathbf{H}}_{mn}^\mathrm{H}\mathbf{C}_{m}^{-1}\mathbf{\Theta}_{mn}^{1/2}\bar{\mathbf{g}}_{mn}\\
 &+\frac{1}{L}c_{mnk}a_{mnk}\hat{\mathbf{g}}_{mnk}^\mathrm{H}\mathbf{\Theta}_{mn}^{1/2}(\mathbf{C}_{m}^{-1}-\mathbf{C}_{mn}^{-1})\bar{\mathbf{H}}_{mn}\mathbf{P}_{mn} \bar{\mathbf{H}}_{mn}^\mathrm{H}\mathbf{C}_{m}^{-1}\mathbf{\Theta}_{mn}^{1/2}\bar{\mathbf{g}}_{mn}.
\end{align}
}
\rule{\textwidth}{0.25mm}
\end{figure*}
Employing Lemma $2$ of \cite{Wagner2012}, $\mathbf{C}_{m}^{-1} - \mathbf{C}_{mn}^{-1} = -\mathbf{C}_{m}^{-1}(\mathbf{C}_{m} - \mathbf{C}_{mn})\mathbf{C}_{mn}^{-1} $ with $\mathbf{C}_{m} - \mathbf{C}_{mn} = \mathbf{\Theta}_{mn}^{1/2} \bar{\mathbf{g}}_{mn} \bar{\mathbf{g}}_{mn}^\mathrm{H} \mathbf{\Theta}_{mn}^{1/2}$. Then, \eqref{determinis} can be written as \eqref{quadratic_form} shown at the top of the next page,
\begin{figure*}
\begin{align}\label{quadratic_form}
\nonumber {\mathbf{h}}_{mnk}^\mathrm{H} \mathbf{\Sigma}_{m} \bar{\mathbf{H}}_{mn} \mathbf{P}_{mn} \bar{\mathbf{H}}_{mn}^\mathrm{H}\mathbf{\Sigma}_{m}{\mathbf{h}}_{mnk}  &= c_{mnk}^2\left(\frac{1}{L}\bar{\mathbf{g}}_{mn}^\mathrm{H} \mathbf{B}_{mn}\bar{\mathbf{g}}_{mn} - \frac{1}{L}\bar{\mathbf{g}}_{mn}^\mathrm{H} \mathbf{A}_{mn} \bar{\mathbf{g}}_{mn}\bar{\mathbf{g}}_{mn}^\mathrm{H} \mathbf{B}_{mn} \bar{\mathbf{g}}_{mn}\right)\\
\nonumber &+ a_{mnk}^2\left(\frac{1}{L}\hat{\mathbf{g}}_{mnk}^\mathrm{H}  \mathbf{B}_{mn}\hat{\mathbf{g}}_{mnk}- \frac{1}{L}\hat{\mathbf{g}}_{mnk}^\mathrm{H}  \mathbf{A}_{mn}\bar{\mathbf{g}}_{mn}\bar{\mathbf{g}}_{mn}^\mathrm{H} \mathbf{B}_{mn} \hat{\mathbf{g}}_{mnk}\right)\\
\nonumber &+ a_{mnk}c_{mnk}\left(\frac{1}{L}\bar{\mathbf{g}}_{mn}^\mathrm{H}\mathbf{B}_{mn} \hat{\mathbf{g}}_{mnk}-\frac{1}{L}\bar{\mathbf{g}}_{mn}^\mathrm{H}\mathbf{A}_{mn}\bar{\mathbf{g}}_{mn} \bar{\mathbf{g}}_{mn}^\mathrm{H} \mathbf{B}_{mn} \hat{\mathbf{g}}_{mnk}\right)\\
 &+ a_{mnk}c_{mnk}\left(\frac{1}{L}\hat{\mathbf{g}}_{mnk}^\mathrm{H}\mathbf{B}_{mn} \bar{\mathbf{g}}_{mn}-\frac{1}{L}\hat{\mathbf{g}}_{mnk}^\mathrm{H}\mathbf{A}_{mn}\bar{\mathbf{g}}_{mn} \bar{\mathbf{g}}_{mn}^\mathrm{H} \mathbf{B}_{mn} \bar{\mathbf{g}}_{mn}\right).
\end{align}
\rule{\textwidth}{0.25mm}
\end{figure*}
where $\mathbf{B}_{mn} = \mathbf{\Theta}_{mn}^{1/2}\mathbf{C}_{mn}^{-1}\bar{\mathbf{H}}_{mn}\mathbf{P}_{mn} \bar{\mathbf{H}}_{mn}^\mathrm{H}\mathbf{C}_{m}^{-1}\mathbf{\Theta}_{mn}^{1/2}$ and $\mathbf{A}_{mn} = \mathbf{\Theta}_{mn}^{1/2}\mathbf{C}_{m}^{-1}\mathbf{\Theta}_{mn}^{1/2}$.
Then, by applying Lemma $7$ of \cite{Wagner2012} to each quadratic form in \eqref{quadratic_form}, we obtain
\begin{eqnarray}\label{inter}
&&\nonumber a. ~ \bar{\mathbf{g}}_{mn}^\mathrm{H} \mathbf{A}_{mn} \bar{\mathbf{g}}_{mn} \mathop \to \limits^{a.s.} \frac{u_{mn} }{1+u_{mn} },\\
&&\nonumber b. ~ \hat{\mathbf{g}}_{mnk}^\mathrm{H} \mathbf{A}_{mn} \bar{\mathbf{g}}_{mn} \mathop \to \limits^{a.s.} 0 ,\\
&& \nonumber c. ~\bar{\mathbf{g}}_{mn}^\mathrm{H} \mathbf{B}_{mn} \bar{\mathbf{g}}_{mn}\mathop \to \limits^{a.s.} \frac{u'_{mn} }{1+u_{mn} },\\
 &&\nonumber d. ~\hat{\mathbf{g}}_{mnk}^\mathrm{H} \mathbf{B}_{mn} \hat{\mathbf{g}}_{mnk}\mathop \to \limits^{a.s.} u'_{mn},\\
&& e. ~\hat{\mathbf{g}}_{mnk}^\mathrm{H} \mathbf{B}_{mn} \bar{\mathbf{g}}_{mn}\mathop \to \limits^{a.s.} 0,
\end{eqnarray}
where $u_{mn}  = \frac{1}{L} {\rm{tr}}[\mathbf{\Theta}_{mn}\mathbf{C}_{mn}^{-1}]$ and $u'_{mn}  =  \frac{1}{L} {\rm{tr}} [\mathbf{P}_{mn} \bar{\mathbf{H}}_{mn}^\mathrm{H}\mathbf{C}_{mn}^{-1}\mathbf{\Theta}_{mn}\mathbf{C}_{mn}^{-1}\bar{\mathbf{H}}_{mn}]$. By substituting \eqref{inter} in \eqref{quadratic_form}, we obtain
\begin{align}\label{quad_eqn}
 \nonumber &{\mathbf{h}}_{mnk}^\mathrm{H} \mathbf{\Sigma}_{m}^{-1} \bar{\mathbf{H}}_{mn} \mathbf{P}_{mn} \bar{\mathbf{H}}_{mn}^\mathrm{H}\mathbf{\Sigma}_{m}^{-1}{\mathbf{h}}_{mnk} \\
 \nonumber &\mathop \to \limits^{a.s.} c_{mnk}^2 \left(  \frac{1}{L} \frac{u'_{mn} }{1+u_{mn} } -  \frac{1}{L} \frac{u_{mn} u'_{mn} }{(1+u_{mn} )^2} \right)   \\
\nonumber  &+a_{mnk}^2  (\frac{1}{L} u'_{mn})\\
  &  \mathop \to \limits^{a.s.} c_{mnk}^2\frac{1}{L} \frac{u'_{mn} }{(1+u_{mn} )^2} +  a_{mnk}^2\frac{1}{L}u'_{mn},
\end{align}
By employing Lemma $6$ of \cite{Wagner2012}, we have
\begin{align}\label{as_u2}
\nonumber u_{mn}  &\mathop \to \limits^{a.s.}  \frac{1}{L} {\rm{tr}} [\mathbf{\Theta}_{mn}\mathbf{C}_{m}^{-1}] \\
\nonumber &\mathop \to \limits^{a.s.} e_{mn}^o \\
\frac{1}{L}u'_{mn}  &\mathop \to \limits^{a.s.} \bar{\Upsilon}_{mn} ,
\end{align}
in which $e_{mn}^o$ is given in \eqref{TT1} and $\bar{\Upsilon}_{mn} =  \frac{1}{L^2}{\rm{tr}}[\mathbf{P}_{mn} \bar{\mathbf{H}}_{mn}^\mathrm{H}\mathbf{C}_{m}^{-1}\mathbf{\Theta}_{mn}\mathbf{C}_{m}^{-1}\bar{\mathbf{H}}_{mn}]$, which can be written as
\begin{align}
\bar{\Upsilon}_{mn} = \frac{1}{L} \sum\limits_{\mystack{n'=1}{n' \ne 
n}}^N{p_{mn'} \bar{\mathbf{g}}_{mn'}^\mathrm{H} \mathbf{\Theta}_{mn'}^{1/2}\mathbf{C}_{m}^{-1}\mathbf{\Theta}_{mn} \mathbf{C}_{m}^{-1}\mathbf{\Theta}_{mn'}^{1/2}\bar{\mathbf{g}}_{mn'}}.
\end{align}
Based on Theorem 2 of \cite{Wagner2012}, we have
\begin{align}
\bar{\Upsilon}_{mn} \mathop \to \limits^{a.s.}\Upsilon_{mn}, 
\end{align}
where $\Upsilon_{mn}$ is given in \eqref{TT4}. 

Inserting \eqref{quad_eqn} and \eqref{as_u2} in \eqref{inter-cluster}, we obtain
\begin{eqnarray}\label{pi2_fin}
\nonumber &\sum\limits_{\mystack{n'=1}{n' \ne 
n}}^N{p_{n'} \sum\limits_{m = 1}^M{\mathbb{E} \left\{ \left|\mathbf{h}_{mnk}^\mathrm{H}\mathbf{w}_{mn'}\right|^2 \right\}}} \\
&\mathop \to \limits^{a.s.}  \sum\limits_{m = 1}^M{ \Upsilon_{mn}\left(\frac{c_{mnk}^2}{(1+e_{mn}^o)^2}+ a_{mnk}^2\right)}.
\end{eqnarray}
\end{itemize}
Finally, by substituting \eqref{Eqn:desired_n}, \eqref{Eqn:gain uncertainty3}, \eqref{SIC_INTER} and \eqref{pi2_fin} in \eqref{Eqn:INTERFERENCE}, $\gamma_{nk}^{\rm mRZF}$ is obtained as \eqref{Eqn:RZF_rate}.  

\bibliographystyle{ieeetr}
\bibliography{ref}

\end{document}